\documentclass[twocolumn]{autart}
\usepackage{epsfig}
\usepackage{epstopdf}
\usepackage{graphicx}
\usepackage[mathscr]{eucal}
\usepackage{amsmath,amsfonts,amssymb,amsbsy,amsopn,amstext}
\usepackage{latexsym}
\usepackage{color,multicol,makeidx}
\usepackage{graphicx}
\usepackage{bm}
\usepackage{mathrsfs}
\usepackage{float}
\usepackage{ccaption}
\usepackage{cite}
\usepackage{natbib}
\usepackage[caption=false,font=normalsize,labelfont=sf,textfont=sf]{subfig}
\usepackage{enumerate}
\usepackage{tabu}
\usepackage{listings}
\def\dref#1{(\ref{#1})}
\def\eq{\displaystyle\stackrel\triangle=}
\def\xra{\xrightarrow}

\newcommand{\diag}{\ensuremath{\mathrm{diag}}}
\DeclareMathOperator*{\rank}{rank} \DeclareMathOperator*{\tr}{Tr}
\DeclareMathOperator*{\argmin}{arg\,min}

\newtheorem{example}{Example}

\newenvironment{proof}[1][Proof]{\noindent\textbf{#1.} }{\hfill \rule{0.5em}{0.5em}}

\begin{document}
	
	\begin{frontmatter}
		
		\title{On Input Design for Regularized LTI System Identification: Power-constrained Input
			\thanksref{footnoteinfo}
		}
		\thanks[footnoteinfo]{This work is
			supported by the Thousand Youth Talents Plan funded by the central
			government of China, the Shenzhen Projects Ji-20170189 and Ji-20160207 funded by the Shenzhen Science and Technology Innovation Council, the
			President's grant under contract No. PF. 01.000249 and the Start-up
			grant under contract No. 2014.0003.23 funded by the Chinese
			University of Hong Kong, Shenzhen, as well as by a research grant
			for junior researchers under contract No. 2014-5894, funded by
			Swedish Research Council.}
\vspace{-3mm}	\author[LIN]{Biqiang Mu}\ead{biqiang.mu@liu.se}
		\ and \author[CUHK]{Tianshi Chen}\ead{tschen@cuhk.edu.cn}
		
		\address[LIN]
		{Division of Automatic Control, Department of Electrical Engineering, Link\"oping University, Link\"oping, 58183, Sweden}
		\address[CUHK]{School of Science and Engineering, The Chinese University of Hong Kong, Shenzhen, 518172, China}
		
		\begin{abstract}
		Input design is an important issue for classical system identification methods but has not been investigated for the kernel-based regularization method (KRM) until very recently.
			In this paper, we consider in the time domain the input design problem of KRMs for LTI system identification. Different from the recent result, we adopt a Bayesian perspective and in particular make use of scalar measures (e.g., the 	$A$-optimality, $D$-optimality, and $E$-optimality) of the Bayesian mean square error matrix as the design criteria subject to power-constraint on the input. Instead to solve the optimization problem directly, we propose a two-step procedure. In the first step, by making suitable assumptions on the unknown input, we construct a quadratic map
	(transformation) of the input such that the transformed input design problems are convex, the number of optimization variables is independent of the number of input data, and their global minima can be
	found efficiently by applying well-developed convex optimization software packages. In the second
	step, we derive the expression of the optimal input based on the global minima found in the first step by solving the inverse image of the quadratic map. In addition, we derive analytic results for some special types of fixed kernels, which provide insights on the input design and
also its dependence on the kernel structure.

		\end{abstract}
		
	\end{frontmatter}

	\section{Introduction}
	
	Over the past few years, the kerel-based regularization method (KRM), which was first
	introduced in
	\cite{Pillonetto2010} and then further developed in
	\cite{Pillonetto2011,Chen2012,Chen2014p}, has received increasing attention in the system identification
	community, see e.g., \cite{Pillonetto2014,Chiuso2016} and the
	references therein. It has become a complement to the
	classical maximum likelihood/prediction error methods
	(ML/PEM),\cite{Ljung1999,Soderstrom1989}, which can be justified in a
	couple of aspects. First, the kernel, through which the
	regularization is defined, provides a carrier for prior knowledge on
	the dynamic system to be identified. Second, the model complexity is
	tuned in a continuous manner through the hyperparameter, which is the
	parameter vector used to parameterize the kernel,
	\cite{Pillonetto2015,Mu2017a}. Third, extensive simulation results
	show that KRM can have better average accuracy and robustness than
	ML/PEM for the data that is short and/or has low signal-to-noise
	ratio, \cite{Chen2012,Pillonetto2014}, and as a result, algorithms
	of KRM have been added to the System Identification Toolbox of
	MATLAB \citep{Ljung2015}.

	Most of the recent progress for KRM focus on the issues of kernel
	design and hyperparameter estimation. For the former issue,
	many kernels have been proposed and analyzed to embed
	various kinds of prior knowledge,
	\cite{Chen2016,Chen2017,Carli2017,Marconato2016,Zorzi2017,Pillonetto2016}. In
	particular, two systematic ways are introduced to design kernels in
	\cite{Chen2017}: one is from a machine learning perspective which
	treats the impulse response as a function, and the other one is from
	a system theory perspective which associates the impulse response
	with a linear time-invariant (LTI) system. In \cite{Zorzi2017}, a
	harmonic analysis is first provided for existing kernels including
	the amplitude modulated locally stationary (AMLS) kernel introduced
	in \cite{Chen2017} and then is shown to be a useful tool to design
	more general kernels. In contrast, there are few results reported
	for the issue of hyperparameter estimation
	\cite{Pillonetto2015,Mu2017a}. In particular, it was shown in
	\cite{Mu2017a} that the Stein's unbiased risk estimator (SURE) is
	asymptotically optimal but the widely used empirical Bayes estimator
	is not.
	
	There are some issues for KRM that have not been addressed adequately including the issue of
	input design. There are numerous results on this issue for ML/PEM,  see e.g., the
	survey papers
	\citep{Mehra1974,Hjalmarsson2005,Gevers2005} and the books
	\citep{Goodwin1977,Ljung1999,Zarrop1977}.
	The current state-of-the-art of input design for ML/PEM, see e.g., \cite{Jansson2005,Hildebrand2003,Hjalmarsson2009}, is to solve the problem in a two-step procedure. The first step is to pose the problem in the frequency domain as a convex optimization problem with a linear matrix inequality constraint and then derive the optimal input power spectrum with respect to certain design criteria, and the second step is to derive the realization of the input corresponding to the optimal power spectrum.
	The typical design criteria are scalar measures (e.g., the trace,
	the determinant or the largest eigenvalue) of the \emph{asymptotic} covariance matrix
	of the parameter estimate or the information matrix of ML/PEM
	subject to various constrains on the inputs (e.g., energy or
	amplitude constraints). The typical realization of the input is the filtered white noise  by spectral factorization of the desired input spectrum \citep{Jansson2005,Hjalmarsson2009} or a multisine signal \citep{Hildebrand2003}.
In contrast, there have been no results reported on this issue for KRM
	until very recently in \cite{Fujimoto2016}, where for a fixed kernel
	(a kernel with fixed hyperparameter), the optimal input is derived by
	maximizing the mutual information between the output and the impulse
	response subject to energy-constraint on the input. The
	proposed method in \cite{Fujimoto2016} is very interesting but
	the number of the optimization variables induced by the input
	design problem is equal to the number of data and thus is expensive to
	solve when the number of data is large. Moreover, the induced
	optimization problem is nonconvex and the proposed gradient-based
	algorithm may be inefficient and subject to local minima
	issue. Nevertheless, their simulation result
	looks quite promising and motivates the interest of further investigation.


	In this paper, we treat the input design problem for KRM from
	a perspective different from \cite{Fujimoto2016}. 	
	Similar to \cite{Fujimoto2016}, we also assume that the kernel is fixed (otherwise, it can be estimated from a preliminary experiment), but
	our starting point is different and is the mean square error (MSE)
	matrix of the regularized finite impulse response (FIR) estimate \cite{Chen2012}. Since the MSE
	matrix depends on the unknown true impulse response, we propose to
	make use of the Bayesian interpretation of the KRM and derive the
	so-called Bayesian MSE matrix, which only depends on the fixed
	kernel and the input. It is then possible to use scalar measures of
	the Bayesian MSE matrix as the design criteria to optimize the
	input, e.g., the $A$-optimality, $D$-optimality, and $E$-optimality
	measures. Interestingly, the design criterion
	in \cite{Fujimoto2016} is equivalent to the
	$D$-optimality of the Bayesian MSE matrix introduced here. Moreover, we also treat the unknown input in a way different from \cite{Fujimoto2016} and consider power-constrained input \cite{Goodwin1977} accordingly.

	Instead to solve the optimization problem induced by the input
	design in a direct way, we propose a two-step procedure. In the
	first step, under the assumption on the unknown input, we construct a quadratic map
	(transformation) of the input such that the transformed input
	design problems are convex, the number of optimization variables is equal to the order of the FIR model,  and their global minima can be
	found efficiently by applying well-developed convex optimization
	software packages, such as CVX \citep{Grant2016}.
	In the second
	step, we derive the expression of the optimal input  based on the global minima found in the first step by	solving the inverse image of the quadratic map.
	From an optimization point of view, a similar optimization problem to the underlying one of our method appeared before in  \cite{Hildebrand2003,Jansson2005}. However,  our method is essentially different. First, our input design problem is posed and solved in the time domain. Second, our method is \emph{not} based on \emph{asymptotic theory} and works for any finite number of data.
Third, the optimal input is derived by solving the inverse image of the quadratic map.
	In addition, we derive analytic results for some special types of fixed kernels, which
	provide insights on input design and
also its dependence on the kernel structure.
 In particular, we show that the impulsive input is globally optimal and the white noise input is asymptotically globally optimal for all diagonal kernels,
	but they are often not optimal for kernels with correlation,
	e.g, the diagonal correlated (DC) kernel introduced in
	\cite{Chen2012}.  Finally, numerical simulation results are provided to illustrate the efficacy of our method.

	The remaining parts of this paper are organized as follows. In
	Section \ref{sec2}, we first review briefly KRM. In Section
	\ref{sec3}, we state the problem formulation of the input design
	problem. In Section \ref{sec:main}, we introduce the two-step
	procedure to solve the input design problem. Then we show in Section
	\ref{sec4} that for some fixed kernels, it is possible to derive the
	explicit solution of the optimal input. The numerical simulation is
	given in Section \ref{sec:sim} to demonstrate the proposed method.
	Finally, we conclude the paper in Section \ref{con}. All proofs of
	Theorems, Propositions, and Lemmas are postponed to the Appendix.

\section{Regularized FIR Model Estimation}
\label{sec2}

Consider a single-input single-output linear stable and causal
discrete-time system
\begin{align}
y(t) = G_0(q^{-1})u(t) + v(t).\label{sys}
\end{align}
Here $t\in\mathbb N$ is the time index, $q^{-1}$ is the backshift
operator ($q^{-1} u(t)=u(t-1)$), $y(t),u(t)\in\mathbb R$ are the
output and input of the system at time $t$, respectively,
$v(t)\in\mathbb R$ is a zero mean white noise with variance
$\sigma^2>0$ and is independent of the input $u(t)$, and the
transfer function $G_0(q^{-1})$ of the ``true'' system is
\begin{align}
G_0(q^{-1}) =\sum_{k=1}^\infty g_k^0 q^{-k}\label{tir},
\end{align}
where the coefficients $g_k^0\in\mathbb R,k=1,\cdots,\infty$ form
the impulse response of the system. Further, we assume that the
input $u(t)$ is known (deterministic). The system identification
problem is to estimate a model of $G_0(q^{-1})$ as well as possible
based on the  data $\{u(t-1),y(t)\}_{t=1}^N$.

The stability of $G_0(q^{-1})$ implies that its impulse response
decays to zero, and thus it is always possible to truncate the
infinite impulse response by a high order finite impulse response
(FIR) model:
\begin{align}
G(q^{-1}) =\sum_{k=1}^n g_k q^{-k},~~\theta=[g_1,\cdots,g_n]^T \in
\mathbb{R}^n, \label{fir}
\end{align}
where $n$ is the order of the FIR model and $(\cdot)^T$ denotes the
transpose of a matrix or vector. With the FIR model (\ref{fir}),
system (\ref{sys}) is written as
\begin{align}
y(t)\!=\!\phi^T\!(t) \theta \!+\! v(t),~\phi(t)\!=\![u(t\!-\!1),\cdots,u(t\!-\!n)]^T, \label{fir2}
\end{align}
and its matrix-vector form is
\begin{align}
&\hspace{3em}Y=\Phi \theta + V,\label{fir3}\\
&Y=[y(1)~y(2)\cdots~y(N)]^T,\\
&\Phi=[\phi(1)~\phi(2)~\cdots~\phi(N)]^T,\label{eq:regressM}\\
&V=[v(1)~v(2)~\cdots~v(N)]^T,
\end{align}
where the unknown input $u(t)$ with
	$t=-1,\dots,-n+1$,  can be handled in different ways: e.g., they can be not used (``non-windowed'')
or they can be set to zero (``pre-windowed''); see \cite[p.
320]{Ljung1999} for discussions.

There are different methods to estimate $\theta$ and the simplest
one is perhaps the least squares (LS) method:
\begin{align}
\widehat{\theta}_N^{\rm LS} = \argmin_{\theta} \|Y-\Phi \theta\|^2
=(\Phi^T\Phi)^{-1}\Phi^TY, \label{ls}
\end{align}
where $\|\cdot\|$ is the Euclidean norm. However,
$\widehat{\theta}_N^{\rm LS}$ may have large variance and thus large
mean square error (MSE), as its variance increases approximately
linearly with respect to $n$. One way to mitigate the possible large
variance and reduce the MSE is by using the regularized least
squares (RLS) method, see e.g., \cite{Chen2012}:
\begin{subequations}\label{eq:rls}
	\begin{align}
	\widehat{\theta}_N^{\rm R}=&\argmin_{\theta \in \mathbb{R}^n}\|Y-\Phi\theta\|^2 + \sigma^2\theta^TP^{-1}\theta\\
	=&P\Phi^T(\Phi P \Phi^T + \sigma^2 I_N)^{-1}Y, \label{rls}
	\end{align}
\end{subequations}
where $P$ is positive semidefinite and is called the kernel matrix
($\sigma^2P^{-1}$ is often called the regularization matrix), and
$I_N$ is the $N$-dimensional identity matrix.


Now we let $\theta_0=[g_1^0,\cdots,g_n^0]^T$,  where
$g_1^0,\dots,g_n^0$ are defined in (\ref{tir}) and we then obtain
the MSE matrix of $\widehat{\theta}_N^{\rm R}$:
\begin{align}
\widehat{\theta}^{\rm R}_N& -\theta_0 =-\sigma^2Q^{-1} P^{-1}\theta_0
 + Q^{-1}\Phi^TV,\\
\nonumber
M_N&=\mathrm{E}(\widehat{\theta}_N^{\rm R}-\theta_0) (\widehat{\theta}_N^{\rm R}-\theta_0) ^T\\
&= \sigma^4 Q^{-1}P^{-1}\theta_0\theta_0^TP^{-1}Q^{-1}
\!+\!\sigma^2Q^{-1} \Phi^T \Phi Q^{-1}\label{msem}\\Q &= \Phi^T  \Phi +
\sigma^2 P^{-1},\label{eq:Q}
\end{align} where $\rm E(\cdot)$ is the mathematical expectation. It
has been shown in \citet[Prop. 2]{Mu2017a} that for a suitably
chosen kernel matrix $P$,  $\tr(M_N(\widehat{\theta}_N^{\rm R}))\leq
\tr(M_N(\widehat{\theta}_N^{\rm LS}))$, where $\tr(\cdot)$ is the
trace of a square matrix.

The problem to achieve a good $\widehat{\theta}^{\rm R}_N$ boils
down to the choice of a suitable kernel matrix $P$, which contains
two issues: kernel design and hyperparameter estimation.

\subsection{Kernel Design}

The issue of kernel design is regarding how to embed in a kernel the
prior knowledge of the underlying system to be identified by
parameterizing the kernel with a parameter vector, say $\eta$, called
hyperparameter. The essence of kernel design is analogous to the
model structure design for ML/PEM, and the kernel
determines the underlying model structure for the regularized FIR
model (\ref{rls}). So far, several kernels have been proposed, such
as the stable spline (SS) kernel \citep{Pillonetto2010}, the
diagonal correlated (DC) kernel and the tuned-correlated (TC) kernel
\citep{Chen2012}, the latter two of which are defined as follows:
\begingroup
\allowdisplaybreaks
\begin{align}
\nonumber
&{\rm DC}:~~P_{kj}(\eta) = c\lambda^{(k+j)/2}\rho^{|j-k|},\\
&~~\eta = [c,\lambda,\rho]\in \Omega=\{c\geq 0,0\leq \lambda \leq 1, |\rho|\leq 1\};\label{kdc}\\
\nonumber
&{\rm TC}:~~P_{kj}(\eta) = c\lambda^{\max(k,j)},\\
&~~~~~~~\eta = [c,\lambda]\in\Omega=\{c\geq 0,0\leq \lambda \leq
1\}.\label{ktc}
\end{align} where the TC kernel \dref{ktc} is a special case of the DC kernel with $\rho=\sqrt{\lambda}$ \citep{Chen2012} and is also called the first order SS kernel \citep{Pillonetto2014}.
\endgroup

\subsection{Hyperparameter Estimation}

Once a kernel is designed, the next step is to determine the
hyperparameter based on the data. The essence of hyperparameter
estimation is analogous to the model order selection for ML/PEM, and the hyperparameter determines the model complexity
of the regularized FIR model (\ref{rls}). Several estimation methods
have been suggested in \cite[Section 14]{Pillonetto2014}. The most
widely used method is the empirical Bayes (also called marginal
likelihood maximization) method. The idea is to adopt the Bayesian
perspective and embed the regularization term
$\sigma^2\theta^TP^{-1}\theta$ in a Bayesian framework. More
specifically, we assume in (\ref{fir2}) that $v(t)$ and $\theta$ are
independent and Gaussian distributed with $v(t) \sim
\mathscr{N}(0,\sigma^2)$ and
\begin{align}\label{bs}
\theta \sim \mathscr{N}(0,P).
\end{align}
Then  $\theta$ and $Y$ are jointly Gaussian distributed and moreover, the posterior distribution of $\theta$ given $Y$ is
\begin{align}
&\theta | Y  \sim \mathscr{N}(\widehat{\theta}^{\rm R},\widehat{P}^{\rm R}),\nonumber\\
&\widehat{\theta}^{\rm R}_N = P\Phi^T(\Phi P \Phi^T + \sigma^2 I_N)^{-1}Y,\nonumber \\
&\widehat{P}^{\rm R}_N = P - P\Phi^T(\Phi P \Phi^T + \sigma^2
I_N)^{-1}\Phi\label{eq:postmean}
P=\sigma^2Q^{-1}.
\end{align}
Moreover, the output $Y$ is also Gaussian with zero mean and
covariance matrix $\Phi P(\eta)\Phi^T+\sigma^2 I_N$. Therefore, the
hyperparameter $\eta$ can be estimated by maximizing the marginal
likelihood, or equivalently,
\begin{align}\label{eq:EB}
\text{EB: }\ \widehat{\eta} 
&=\argmin_{\eta\in\Omega} Y^T F(\eta)^{-1}Y +\log \det (F(\eta)),\\
F&=\Phi^TP\Phi+\sigma^2 I_N,\label{eq:F}
\end{align} where $\det(\cdot)$ is the determinant of a matrix.
The EB (\ref{eq:EB}) has the advantage that it is robust
\cite{Pillonetto2015}, but not asymptotical optimal in the
sense of MSE \citep{Mu2017a}.


\section{Problem Formulation for Input Design}
\label{sec3}

\subsection{Bayesian MSE Matrix and Design Criteria}

The goal of input design is to determine, under suitable conditions,
an input sequence such that the regularized FIR model \dref{rls} is
as good as possible. The MSE matrix \dref{msem} is a natural measure
to evaluate how good the regularized FIR model estimate \dref{rls} is.
Unfortunately, the first term of \dref{msem} depends on the true
impulse response $\theta_0$ and thus it can not be used directly.

There are different ways to deal with this difficulty. We adopt
a Bayesian perspective that is similar to the derivation of
(\ref{eq:EB}). More specifically, we assume that \begin{align}\label{eq:ass}
\theta_0 \sim \mathscr{N}(0,P).\end{align} Then taking expectation
on both side of (\ref{msem}) leads to \begin{align} \nonumber  M_N
&=\sigma^4 Q^{-1}P^{-1}Q^{-1}
+\sigma^2Q^{-1} \Phi^T \Phi Q^{-1}\\
\nonumber
&=\sigma^2 Q^{-1} (\sigma^2P^{-1}+ \Phi^T \Phi)Q^{-1}\\
&=\sigma^2 Q^{-1}, \text{\ ``Bayesian MSE''}. \label{bmse}
\end{align}
Since the Bayesian perspective is adopted, (\ref{bmse}) is called
the Bayesian MSE matrix of the regularized FIR model estimate
(\ref{rls}) under the assumption \eqref{eq:ass}. Interestingly, the
Bayesian MSE matrix \dref{bmse} is equal to the posterior covariance
of the regularized FIR model estimate (\ref{eq:postmean}) under the
assumption \eqref{bs}.

We will tackle the input design problem by minimizing a scalar
measure of the Bayesian MSE matrix (\ref{bmse}) of the regularized
FIR model estimate (\ref{rls}). This idea is similar to that of the
traditional input design problem by minimizing a scalar measure of
the asymptotic covariance matrix of the parameter estimate
\citep{Ljung1999}.
The typical $A$-optimality, $D$-optimality, and $E$-optimality
scalar measures of the Bayesian MSE matrix (\ref{bmse}) will be
chosen as the design criteria for the input design problem and given
below:
\begin{align}
&D-\mbox{optimality}: ~~\det (M_N)\eq\det (\sigma^2 Q^{-1}), \label{dcd}\\
&A-\mbox{optimality}: ~~{\rm Tr}(M_N)\eq{\rm Tr}(\sigma^2 Q^{-1}),\label{dca}\\
&E-\mbox{optimality}: ~~\lambda_{\rm max} (  M_N)\eq \lambda_{\rm
	max} (\sigma^2 Q^{-1}),\label{dce}
\end{align} where $\lambda_{\rm max}(\cdot)$ is the largest
eigenvalue of a matrix.

\begin{rem}
	Another possible way to deal with the unknown $\theta_0$ in the MSE
	matrix \dref{msem} is to replace $\theta_0$ by an estimate obtained
	in advance, for example, the LS estimate \dref{ls}. In this case,
	the MSE matrix \dref{msem} becomes
	\begin{align}
	\sigma^4 Q^{-1}P^{-1}\widehat{\theta}_N^{\rm LS}
	(\widehat{\theta}_N^{\rm LS})^TP^{-1}Q^{-1} +\sigma^2R^{-1} \Phi^T
	\Phi Q^{-1}. \label{lsmse}
	\end{align} Clearly, the input design would depend on the
	quality of the chosen estimate.
	
\end{rem}

\begin{rem} Interestingly,  the $D$-optimality measure \dref{dcd} is
	equivalent to the mutual information between the
	true impulse response and the output $Y$ given in
	\cite{Fujimoto2016}.
\end{rem}

\subsection{Problem Statement}

Now we are able to state the input design problem based on
the design criteria \dref{dcd}--\dref{dce}.

Assume that the kernel matrix $P(\eta)$ and the variance $\sigma^2$ of the
noise  are known in advance (otherwise, they can be estimated with a preliminary experiment) and also note that the construction of the regression matrix $\Phi$ in
(\ref{eq:regressM}) requires the inputs $u_{-n+1},\cdots,u_{N-1}$. Then the input design problem is to determine an input sequence
\begin{align}\label{eq:inp_sequence_full}
u_{-n+1},\cdots,u_{-1},u_0,u_1,\cdots,u_{N-1},\end{align} where, for simplicity, $u_t$ is used to denote $u(t)$  hereafter, such that the input sequence (\ref{eq:inp_sequence_full}) minimizes the design criteria \dref{dcd}--\dref{dce} subject to certain constraints.

It should be noted that the inputs $u_{-n+1},\cdots,u_{-1}$ are unknown for the identification problem but can be treated as design variables for the input design problem. In particular, \cite{Fujimoto2016} 
sets the inputs $u_{-n+1},\cdots,u_{-1}$ to be zero, i.e., \begin{align}\label{eq:FS_assoninicond}
  u_{-n+1}=\cdots=u_{-1}=0,
\end{align} 
and essentially minimizes the $D$-optimality measure \dref{dcd} subject to an energy-constraint 
\begin{equation}
\begin{aligned}
\mathcal{U}&=\{u|u^Tu\leq \mathscr{E},u\in\mathbb{R}^{N}\},\\
 u&=[u_0,u_1,\cdots,u_{N-1}]^T,
\label{ic}
\end{aligned}
\end{equation}
where $\mathscr{E}>0$ is a known constant and is the
maximum available energy for the input.

Here, we treat $u_{-n+1},\cdots,u_{-1}$ in a different way. Our idea is not only to reduce the number of optimization variables of the input design problem, but also to bring it certain structure such that it becomes easier to solve, see Section \ref{sec:main} for details.  In
particular, we assume $N\geq n$ and set
\begin{equation}
\begin{aligned}\label{eq:assoninicond}
u_{-i} &= u_{N-i},\quad i=1,\cdots,n-1.
\end{aligned}\end{equation}
The advantage of doing so is that the regression matrix $\Phi$ has
a good structure and becomes a circulant matrix:
\begin{align*}
\Phi=\left[
\begin{array}{cccccc}
u_0&u_{N-1}&\cdots &u_{N-n+2} &u_{N-n+1}\\
u_1&u_0&\cdots &u_{N-n+3}&u_{N-n+2}\\
\vdots&\vdots&\ddots&\vdots&\vdots\\
u_{n-2}&u_{n-3}&\cdots&u_0&u_{N-1}\\
u_{n-1}&u_{n-2}&\cdots&u_1&u_0\\
u_{n}&u_{n-1}&\cdots&u_2&u_1\\
\vdots&\vdots&\ddots&\vdots&\vdots\\
u_{N-1}&u_{N-2}&\cdots&u_{N-n+1}&u_{N-n}
\end{array}
\right].
\end{align*}
Moreover, we consider the power constraint
on the input sequence (\ref{eq:inp_sequence_full}) used for the traditional ML/PEM input design problem, see e.g.,
\cite[p. 129, eq. (6.3.12)]{Goodwin1977}, i.e., 
\begin{align}\label{eq:constraint}
\sum_{t=1}^N u_{t-i}^2=\mathscr{E}, \ i=1,\dots,n,
\end{align} 
Under the assumption (\ref{eq:assoninicond}), the input design problem of minimizing the design criteria \dref{dcd}--\dref{dce} subject to the constraint (\ref{eq:constraint}) can be equivalently written as follows:
\begingroup
\allowdisplaybreaks
\begin{align}
&u^*_{\rm D}\eq\arg\min_{u\in \mathscr{ U}} {\rm \det} (\sigma^2 Q^{-1}), \label{dop_o}\\
&u^*_{\rm A}\eq\arg\min_{u\in \mathscr{ U}} {\rm Tr} (\sigma^2 Q^{-1}), \label{aop_o}\\
&u^*_{\rm E}\eq\arg\min_{u\in \mathscr{ U}} \lambda_{\rm max}
(\sigma^2 Q^{-1}), \label{eop_o}
\end{align}
\endgroup
where the constraint (\ref{eq:constraint}) under the assumption (\ref{eq:assoninicond}) becomes
\begin{equation}
\begin{aligned}
\mathscr{U}&=\{u|u^Tu= \mathscr{E},u\in\mathbb{R}^{N}\},\\
u&=[u_0,u_1,\cdots,u_{N-1}]^T,
\label{ic_bar}
\end{aligned}
\end{equation}
It is worth to note that the optimal input is not unique in general.
To check this, note that if an input sequence $u^*$ is optimal, then
$-u^*$ is also optimal, where $u^*$ can be any of $u^*_{\rm
	D},u^*_{\rm A}$ and $u^*_{\rm E}$. Therefore, $u^*_{\rm D},u^*_{\rm
	A},u^*_{\rm E}$ should be understood as the set of all optimal
inputs minimizing \dref{dop_o}, \dref{aop_o}, \dref{eop_o}, respectively.

%
%
%
%
%
%

\section{Main Results}

\label{sec:main}

Similar to \cite{Fujimoto2016}, we can also try to solve the input
design problems \dref{dop_o}--\dref{eop_o} by using gradient-based
algorithms. However, such algorithms may have the following problems:
\begin{enumerate}[1)]
	\item the input design
	problems \dref{dop_o}--\dref{eop_o} is expensive to solve for large $N$,
	because the number of optimization variables is $N$.
	
	\item the gradient-based algorithms may be subject to the issue of local
	minima, because the input design problems
	\dref{dop_o}--\dref{eop_o} are non-convex.

\end{enumerate}

We propose to solve the input design problems
\dref{dop_o}--\dref{eop_o} in a two-step procedure, and its idea is
sketched below:
\begin{enumerate}[1)]
	\item under the assumption (\ref{eq:assoninicond}), we construct a quadratic map
	(transformation) of $u$, say $f(u)$, such that the transformed input
	design problems of \dref{dop_o}--\dref{eop_o} are convex and
	thus can be solved efficiently.

	\item by exploiting the properties of $f(\cdot)$ and based on the globally minima of the transformed input design problems of
	\dref{dop_o}--\dref{eop_o}, we  derive the optimal inputs
	by solving the inverse of the map $f(\cdot)$.
\end{enumerate}

\subsection{A Quadratic Map}

Under the assumption (\ref{eq:assoninicond}), we define
\begin{equation}
\begin{aligned}
\label{rd} r&=[r_0,r_1,\cdots,r_{n-1}]^T,\\
r_j&=\sum_{k=0}^{N-1}u_{k}u_{k-j},\ j=0,\dots, n-1.
\end{aligned}
\end{equation}
Then we have
\begin{align}
R\eq\Phi^T\Phi =\left[
\begin{array}{cccccc}
r_0&r_{1}&\cdots &r_{n-2} &r_{n-1}\\
r_1&r_0&\ddots &r_{n-3}&r_{n-2}\\
\vdots&\ddots&\ddots&\ddots&\vdots\\
r_{n-2}&r_{n-3}&\ddots&r_0&r_{1}\\
r_{n-1}&r_{n-2}&\cdots&r_1&r_0\\
\end{array}
\right],\label{phiphi}
\end{align}
which is a positive semidefinite Toeplitz matrix.

Actually, the definition \dref{rd} determines a quadratic
vector-valued function
\begin{align}\label{eq:defmap}r=f(u)=[f_0(u),\cdots,f_{n-1}(u)]^T.
\end{align}
The domain of $f(\cdot)$ is
$\mathscr{U}$ which has been defined in \eqref{ic_bar} and is convex and compact,
and the $j$-element of $f(\cdot)$ is
\begin{align}
&f_j(u)=u^TL_ju,\  j=0,\dots, n-1,\\ &L_j= \frac12\left[
\begin{array}{cc}
0&I_{N-j}\\
I_{j}&0
\end{array}
\right] + \frac12\left[
\begin{array}{cc}
0&I_{j}\\
I_{N-j}&0
\end{array}
\right].
\end{align}
It follows that the corresponding value $r_0$ is $\mathscr{E}$ for any input $u\in \mathscr{U}$.
Moreover, denote the image of $f(\cdot)$ under $\mathscr{U}$ by
\begin{align}
\mathscr{F}=\{f(u)|u\in\mathscr{U}\},\label{fi}
\end{align}
which is a convex polytope (See Theorem \ref{thm2a} below).

\subsection{Transformed Input Design Problem}

Interestingly, the quadratic map (\ref{eq:defmap}) makes the
transformed input design problems of \dref{dop_o} to \dref{eop_o} convex.

To show this, note that the matrix $Q$ in (\ref{eq:Q}) can be
written as
\begin{align}
Q(r)=r_0I_n+r_1Q_1+\cdots+r_{n-1}Q_{r-1}+\sigma^2P^{-1},
\end{align}
where $Q_i\in\mathbb R^{n\times n}$, $i=1,\dots, n-1$ is symmetric
and zero everywhere except the $(jl)$-element which is one if
$|j-l|=i$ with $j,l=1,\dots,n$. The map (\ref{eq:defmap}) transforms
the input design problems \dref{dop_o}--\dref{eop_o} to the following
optimization problems, respectively,
\begin{align}
&r^*_{\rm D} =\argmin_{r\in \mathscr{F}} \log\det(\sigma^2 Q(r)^{-1})\label{odr}\\
&r^*_{\rm A}=\argmin_{r\in \mathscr{F}} {\rm Tr}(\sigma^2 Q(r)^{-1})\label{oar}\\
&r^*_{\rm E}=\argmin_{r\in \mathscr{F}} \lambda_{\rm max} (\sigma^2 Q(r)^{-1})\label{oer}
\end{align}
where $\log\det(\cdot)$ is used instead of $\det(\cdot)$.

\begin{prop}\label{prop1} Consider the transformed input design problems
	\dref{odr} to \dref{oer}. It follows that
	\begin{itemize}
		\item The problems \dref{odr} and \dref{oar} are
		strictly convex and thus $r^*_{\rm D}$ and $r^*_{\rm A}$ are unique
		in $\mathscr{F}$.

		\item The problem \dref{oer} is convex and $r^*_{\rm E}$ exists but may
		be not unique in $\mathscr{F}$.
		
	\end{itemize}

\end{prop}

\begin{rem}
From an optimization point of view, the optimization problems \dref{odr} to \dref{oer}, modulo the constraint, appeared before, see e.g., \cite{Jansson2005}. However, the contexts are different: the elements of $r$ are the auto-correlation coefficients of the input power spectrum in \cite{Jansson2005} but they do not have such interpretation here unless we divide them by $N$ and let $N$ go to $\infty$.
In addition, the positivity constraint of the input spectrum is transformed into a linear matrix inequality by applying the Kalman-Yakubovich-Popov lemma in \cite{Jansson2005},
	while the constraint here is the feasible set $\mathscr{F}$, which is a convex polytope (See Remark \ref{rem6} below).
\end{rem}

\begin{rem}
	An optimal $r^*$ satisfies $r_0^*=\mathscr{E}$, where
	$r_0^*$ is the first element of $r^*$, and $r^*$ could be $r^*_{\rm D}$ $r^*_{\rm A}$, and $r^*_{\rm E}$.
	
\end{rem}

\begin{rem}
For $ i=1,\dots, n-1$, the elements of the gradient of \dref{odr} and \dref{oar} are respectively given by
	\begin{align}
	&\frac{\partial \log\det(\sigma^2 Q(r)^{-1})}{\partial r_i}
	\! =\! -{\rm Tr} (Q(r)^{-1}Q_i),\label{gvd}\\
	&\frac{\partial{\rm Tr}( \sigma^2 Q(r)^{-1})}{\partial r_i} = - \sigma^2{\rm Tr} (Q(r)^{-2}Q_i).\label{gva}
	\end{align}
	
\end{rem}

\subsection{Finding Optimal Inputs }

Now we assume that the global minima of the transformed input
design problems \dref{odr} to \dref{oer} have been found and denoted
by $r^*\in \mathscr{F}$, which could be any one of $r^*_{\rm
	D},r^*_{\rm A},$ and $r^*_{\rm E}$. Then we consider the problem how
to derive the optimal input $u^*$ from $r^*\in \mathscr{F}$ by exploiting the properties of the quadratic map $f(\cdot)$ and the inverse map of $f(\cdot)$ for a given vector $r\in
\mathscr{F}$.

\subsubsection{Properties of the Quadratic Map}

The map $f(\cdot)$ has the following properties, which ease the derivation of the optimal input
as will be seen shortly.
\begin{thm}
	\label{thm1} Consider the map $f(\cdot)$ defined in \eqref{eq:defmap}. Then there exists a matrix $S\in{\mathbb R}^{n\times N}$ and an orthogonal matrix $W\in{\mathbb R}^{N\times N}$ such that the map $f(\cdot)$ can be written in a composite form as follows:
	\begin{align}
	f(u) &= h_1(h_2(h_3(u)))~~\mbox{with} \label{m1}\\
	h_1(x)&=Sx\\
	h_2(z)&=[z_0^2,z_1^2,\cdots,z_{N-1}^2]^T\\
	h_3(u) &= W^T u \label{m4}
	\end{align}
	where $x=[x_0,x_1,\cdots,x_{N-1}]^T$, $z=[z_0,z_1,\cdots,z_{N-1}]^T$,
	Moreover, the image  of $h_3(\cdot)$ is
	$$\mathscr{Z}=\{h_3(u)|u\in \mathscr{U}\}=\{z|z^Tz= \mathscr{E}\},$$
	the image  of $h_2(\cdot)$
	\begin{align}
	\mathscr{X}&=\{h_2(z)|z\in \mathscr{Z}\}\\
	&=\left\{x\Big|\sum_{i=0}^{N-1}x_i\! =\! \mathscr{E},x_i \geq 0,i\!=\!0,1,\dots,N\!-\!1\right\},\label{ih2}
	\end{align}
	is a convex polytope
	and the image of $h_1(\cdot)$
	\begin{align}
	\mathscr{F}=\{f(u)|u\in\mathscr{U}\} =\{Sx|x\in \mathscr{X}\}
	\end{align}
	is also a convex polytope.
\end{thm}

One choice for the matrices $S$ and $W$ in Theorem \ref{thm1} is
given in the following theorem.

\begin{thm}
	\label{thm2a}
	Let $\varpi=2\pi/N$ and define the vectors
	\begin{align*}
	&\xi_j= \left[
	\begin{array}{c}
	1\\
	\cos(j\varpi)\\
	\cos(2j\varpi)\\
	\vdots\\
	\cos((N\!\!-\!1)j\varpi)
	\end{array}
	\right],~ \zeta_j= \left[
	\begin{array}{c}
	0\\
	\sin(j\varpi)\\
	\sin(2j\varpi)\\
	\vdots\\
	\sin((N\!\!-\!1)j\varpi)
	\end{array}
	\right]
	\end{align*}
	for $j=0,1,\dots$. Then, one choice for the matrices $W$ and $S$ in
	Theorem \ref{thm1} is
	
	1) when $N$ is even,
	\begin{align}
	&W\!=\!\sqrt{\frac{2}{N}} \left[
	\frac{\xi_0}{\sqrt{2}},\xi_1,\cdots,\xi_{\frac{N-2}{2}},\frac{\xi_{\frac{N}{2}}}{\sqrt{2}},\zeta_{\frac{N-2}{2}},\cdots,\zeta_1
	\right]\label{we}\\
	\nonumber
	&S=\Big[\xi_0,\xi_1,\cdots,\xi_{n-1}\Big]^T\\
	&\hspace{0.6em}=\left[\xi_0(1\!:\!n),\xi_1(1\!:\!n),\cdots,\xi_\frac{N}{2}(1\!:\!n),\cdots,
	\xi_1(1\!:\!n)\right]\!\!\!\!\label{seven}
	\end{align}
	and $\rank(S)=\min(N/2\!+\!1,n)$. Also,
	\begin{align}
	\mathscr{F}\!=\!\mathbf{conv}\left\{\mathscr{E}\xi_0(1\!:\!n),\mathscr{E}\xi_1(1\!:\!n),\cdots,\mathscr{E}\xi_\frac{N}{2}(1\!:\!n)
	\right\}\label{conhe}
	\end{align}
	%
	where $\mathbf{conv}\{\cdots\}$ is the convex hull of the corresponding vectors and $\xi_j(1\!:\!n)$ is the vector consisting of the first $n$
	elements of $\xi_j$.

	2) when $N$ is odd,
	\begin{align}
	&W=\sqrt{\frac{2}{N}} \left[
	\frac{\xi_0}{\sqrt{2}},\xi_1,\cdots,\xi_{\frac{N-1}{2}},\zeta_{\frac{N-1}{2}},\cdots,\zeta_1
	\right]\\
	\nonumber
	&S=\Big[\xi_0,\xi_1,\cdots,\xi_{n-1}\Big]^T\\
	\nonumber
	&\hspace{0.7em}=\left[\xi_0(1\!:\!n),\xi_1(1\!:\!n),\cdots,\xi_\frac{N-1}{2}(1\!:\!n),\right.\\
	&\hspace{7em}\left.\xi_\frac{N-1}{2}(1\!:\!n),\cdots,
	\xi_1(1\!:\!n)\right]\!\!\!\!\label{sodd}
	\end{align}
	and $\rank(S)=\min((N\!+\!1)/2,n)$. In addition,
	\begin{align}
	\mathscr{F}\!=\!\mathbf{conv}\left\{\mathscr{E}\xi_0(1\!:\!n),\mathscr{E}\xi_1(1\!:\!n),\cdots,\mathscr{E}\xi_\frac{N\!-\!1}{2}(1\!:\!n)
	\right\}.\label{conho}
	\end{align}
	
\end{thm}


\begin{rem}
	\label{rem6}
	The transformed input design problems \dref{odr}, \dref{oar},
	\dref{oer} can be solved effectively by CVX \citep{Grant2016}. The
	optimization criteria of \dref{odr}, \dref{oar}, \dref{oer} are
	standard in CVX and we only need to rewrite the feasible set
	$\mathscr{F}$ as a number of linear equalities and inequalities by the definition of convex hull \dref{conhe} or \dref{conho}:
	\begin{itemize}
		\item when $N$ is even, we have from \dref{conhe} that
		\begin{align}\label{fsc1}
		\begin{array}{l}
		\mathscr{E}S(:,1\!:\!N/2\!+\!1) a=r\\
		a=[a_1,a_2,\cdots,a_{N/2+1}]^T,\ a_j\geq 0
		\end{array}
		\end{align}
		where $r\in\mathscr{F}$, $S(:,1\!:\!N/2\!+\!1)$ denotes the first $N/2\!+\!1$ columns of $S$,
		and the first equality constraint is $\sum_{j=1}^{N/2+1} a_j =1$ since the first element of $r$ is $\mathscr{E}$ and all the elements of the first row of $S$ are one.
		\item when $N$ is odd, we have from \dref{conho} that
		\begin{align}\label{fsc2}
		\begin{array}{l}
		\mathscr{E}S(:,1\!:\!(N\!+\!1)/2) a=r\\
		a=[a_1,a_2,\cdots,a_{(N+1)/2}]^T,\ a_j\geq 0.
		\end{array}
		\end{align}
		
	\end{itemize}
	Since the columns of $S$ are symmetric (See \dref{seven} and
	\dref{sodd}),  $\mathscr{F}$ can be uniformly expressed by for even and odd $N$
	\begin{align}\label{fsc3}
	\begin{array}{l}
	\mathscr{E}S a=r\\
	a=[a_1,a_2,\cdots,a_N]^T,\ a_j\geq 0.
	\end{array}
	\end{align}
	
\end{rem}

\subsubsection{Derivation of the Optimal Input}

Given any $r\in \mathscr{F}$, the inverse image $f^{-1}(r)$ of
$f(\cdot)$ can be derived in the following way based on \dref{m1}--\dref{m4}:
\begin{enumerate}[1)]
	\item we find the inverse image of $h_1(\cdot)$ for $r\in \mathscr{F}$:
	\begin{align}
	\nonumber
	\mathscr{X}(r)
	&\eq\left\{x |Sx =r,x \in\mathscr{X} \right\}\\
	&=\{\mathscr{K}(S) \oplus S^{\dagger}r\} \cap \mathscr{X}
	\label{a1}
	\end{align}
	where $S^{\dagger}$ is the Moore-Penrose pseudoinverse of $S$
	and $\mathscr{K}(S)=\{x|Sx=0\}$ is the null space of $S$, and $\oplus$ is the direct sum between two subspace. 
	
	\item we find the inverse image of $h_2(\cdot)$ for $x\in\mathscr{X}(r)$:
	\begin{align}
	\nonumber
	\mathscr{Z}(r)
	&\!\eq\!\left\{z|h_2(z) \in\mathscr{X}(r)\right\}\\
	&\!=\!\left\{[\pm \sqrt{x_0},\cdots,\pm \sqrt{x_{N-1}}]^T|x\in\mathscr{X}(r) \right\}.\label{a2}
	\end{align}
	
	\item we find the inverse image of $h_3(\cdot)$ for $z\in \mathscr{Z}(r)$:
	\begin{align}
	\mathscr{U}(r)
	\eq\left\{u|W^Tu\in\mathscr{Z}(r)\right\}=\left\{ Wz|z\in \mathscr{Z}(r)\right\} .\label{a3}
	\end{align}
\end{enumerate}

It is worth to note that $\mathscr{X}(r)$ is still a convex
polytope.

\begin{prop}\label{prop2}
	The inverse image
	$\mathscr{X}(r)$
	is a convex polytope
	and is determined by a group of halfspaces and hyperplanes:
	\begin{align}
	\mathscr{X}(r)=\left\{x|Sx = r, x_i\geq 0,~ i=0,\dots, N\!-\!1   \right\}.
	\end{align}
\end{prop}
This means that the optimal inputs $u^*=\mathscr{U}(r^*)$, where $u^*$ can be any of $u^*_{\rm
	D},u^*_{\rm A}$ and $u^*_{\rm E}$.

\begin{rem}
	Noting that $W$ is orthogonal and the row vectors of $S$ are part of the column vectors of $W$ subject to some scaling factors, a basis of $\mathscr{K}(S)$ can be derived as follows:
	\begin{itemize}
		\item when $N$ is even, a basis of $\mathscr{K}(S)$ is
		\begin{align*}
		\left\{
		\begin{array}{ll}
		\{\zeta_1,\cdots,\zeta_{\frac{N-2}{2}}\}&\mbox{if}~n\!-\!1\geq N/2,\\
		\{\xi_{n},\cdots,\xi_{\frac{N}{2}},\zeta_1,\cdots,\zeta_{\frac{N-2}{2}}\}
		&\mbox{if}~n\!-\!1< N/2.
		\end{array}
		\right.
		\end{align*}
		
		\item when $N$ is odd, a basis of $\mathscr{K}(S)$ is
		\begin{align*}
		\left\{
		\begin{array}{ll}
		\{\zeta_1,\cdots,\zeta_{\frac{N-1}{2}}\}&\mbox{if}~n\!-\!1\geq (N\!-\!1)/2,\\
		\{\xi_{n},\cdots,\xi_{\frac{N\!-\!1}{2}},\zeta_1,\cdots,\zeta_{\frac{N-2}{2}}\}
		&\mbox{if}~n\!-\!1< (N\!-\!1)/2.
		\end{array}
		\right.
		\end{align*}
	\end{itemize}
\end{rem}

\begin{rem}
	As can be seen from
	the algorithm given above for finding the optimal inputs $u^*$ from $r^*\in\mathscr{F}$, the key is to determine the set $\mathscr{X}(r^*)$ while steps 2) and 3) are straightforward.
	Actually, some elements belonging to $\mathscr{X}(r^*)$ can be derived from the global minima  of the optimization problems \dref{odr}--\dref{oer}.
	Assume that the global minima corresponding to the constraints \dref{fsc1} is $\{r^*,a^*\}$, respectively, where $a^*=[a_1^*,a_2^*,\cdots,a_{N/2+1}^*]^T$.
	Then all vectors
	$$\mathscr{E}\Big[a_1^*,\alpha_2a_2^*,\cdots,
	\alpha_{\frac N2}a_{\frac N2}^*,a_{\frac N2\!+\!1}^*,
	(1-\alpha_{\frac N2})a_{\frac N2}^*,\cdots,(1-\alpha_{2})a_{2}^*\Big]^T$$
	belong to $\mathscr{X}(r^*)$ for $0\leq \alpha_i\leq 1,i=2,\cdots,N/2$.
	Accordingly, for the constraint \dref{fsc2} with $a^*=[a_1^*,a_2^*,\cdots,a_{\frac{N+1}2}^*]^T$, we have all vectors
	$$\mathscr{E}\Big[a_1^*,\alpha_2a_2^*,\cdots,
	\alpha_{\frac {N+1}2}a_{\frac {N+1}2}^*,
	(1-\alpha_{\frac {N+1}2})a_{\frac {N+1}2}^*,\cdots,(1-\alpha_{2})a_{2}^*\Big]^T$$
	belong to $\mathscr{X}(r^*)$ for $0\leq \alpha_i\leq 1,i=2,\cdots,(N\!+\!1)/2$.
	At last, for the constraint \dref{fsc3} with $a^*=[a_1^*,a_2^*,\cdots,a_N^*]^T$, we have the vector $\mathscr{E}a^*$ belongs to $\mathscr{X}(r^*)$.
\end{rem}

\section{Optimal Input for Some Special Cases}
\label{sec4}

In general, there is no analytic solution to the input design
problems. In this section, we study the optimal input for
some special types of fixed kernels. The obtained analytic results provide insights on input design and
also its dependence on the kernel structure.



\subsection{Diagonal Kernel Matrices}
We first consider the ridge kernel matrix and then more general
diagonal kernel matrices.
\subsubsection{Ridge Kernel Matrix}
When the kernel matrix is
\begin{align}
P=c I_n,~\mbox{with}~c > 0, \label{rk}
\end{align}
it is possible to obtain the explicit expression of $r^*_{\rm D}$
and $r^*_{\rm A}$, which is given in the following proposition.
\begin{prop}
	\label{thm2}
	Consider the ridge kernel matrix \dref{rk}.
	Then we have
	\begin{align}
	r^*_{\rm D}=r^*_{\rm
		A}=[\mathscr{E},\underbrace{0,\cdots,0}_{n-1~\mbox{\rm zeros}}]^T\eq r^\dagger.\label{dopc}
	\end{align}
\end{prop}

Proposition \ref{thm2} shows that for ridge kernel matrix (\ref{rk}), any input such that $\Phi^T\Phi=\mathscr{E}I_n$ is the global minimum of \eqref{dop_o} and \eqref{aop_o}.
The optimal input given in Proposition \ref{thm2} is identical with the classic result of input design for FIR model estimation with ML/PEM given in \cite{Goodwin1977}.
In fact, this result is not surprising since the ridge kernel assumes that
the prior distribution of the true impulse response is independent and identically distributed Gaussian, which does not provide any information on the correlation and the decay rate of the true impulse response.
The optimal solution \dref{dopc} implies that the optimal inputs could be:
\begin{enumerate}[1)]
	\item Impulsive inputs: $[\pm\sqrt{\mathscr{E}},0,\cdots,0]^T$;
	\item White noise inputs as $N\rightarrow \infty$.
\end{enumerate}
Note that the impulsive input is globally optimal since it satisfies \dref{dopc} exactly. In contrast, it was only shown in \cite{Fujimoto2016} that the impulsive input is \emph{locally} optimal.
In addition, as $N\rightarrow\infty$, the white noise input is also globally optimal. Clearly, the white noise input is \emph{approximately} globally optimal for large $N$.

\subsubsection{General Diagonal Kernel Matrices}
When the kernel matrix is
\begin{align}
P=\diag([\lambda_1,\lambda_2,\cdots,\lambda_n])\label{dk}
\end{align}
where $\lambda_i>0$, $i=1,\dots,n$, and $\diag(\cdot)$ represents a diagonal matrix, the optimal $r^*_{\rm D}$ and
$r^*_{\rm A}$ are same as the one in Proposition \ref{thm2}.
\begin{thm}
	\label{thm3}
	Consider the diagonal kernel matrix \dref{dk}. Then
	we have
	\begin{align}
	r^*_{\rm D}=r^*_{\rm
		A}=r^\dagger.\label{dopcdk}
	\end{align}
\end{thm}

Note that the diagonal kernel matrix \dref{dk} contains the DI
kernel matrix studied in \cite{Chen2012} as a special case, which
takes the following form
$$P=\diag(c[\lambda,\lambda^2,\cdots,\lambda^n]),\ c,\lambda>0.$$
Theorem \ref{thm3} shows that for general diagonal kernel matrices
the information on the decay rate of the impulse response is not
useful for the input design, if no information on the correlation of
the true impulse response is provided.

\subsection{Nondiagonal Kernel Matrices}

Theorem \ref{thm3} motivates an interesting question that, for
nondiagonal kernel matrices, whether $r^*_{\rm D}$ and $r^*_{\rm A}$
are still equal to $r^\dagger$ or not.

\subsubsection{Kernel matrix with a tridiagonal inverse}

To answer this question, we first consider a class of kernel matrices $P$ such that it is
nonsingular and its inverse is tridiagonal, i.e., $P^{-1}$ can be
written as follows:
\begin{align}
P^{-1}=\left[
\begin{array}{ccccc}
p_1&-e_2&&&\\
-e_2&p_2&-e_3&&\\
&\ddots&\ddots&\ddots&\\
&&-e_{n-1}&p_{n-1}&-e_n\\
&&&-e_n&p_n
\end{array}
\right]\label{kti}
\end{align}
where $p_i>0$ for $1\leq i\leq n$. Then, we have the following result.
\begin{thm}
	\label{thm6}
	Consider a kernel matrix $P$ which has tridiagonal inverse \dref{kti}.
	Then we have
	\begin{align}
	r^*_{\rm D}\neq r^\dagger~\text{and }~
	r^*_{\rm A}\neq r^\dagger\label{ndki}
	\end{align}
	for the following two cases:
	\begin{enumerate}[1)]
		\item $e_j>0$ for $j=2,\dots, n$ with any $N$;
		\item $e_j<0$ for $j=2,\dots, n$ with  even $N$.
	\end{enumerate}
\end{thm}
One may wonder whether there exist kernel matrices such that its
inverse is tridiagonal. The answer is positive. Actually, the DC
kernel (\ref{kdc}) has tridiagonal inverse \cite{Carli2017}, and
moreover, the simulation induced kernels with Markov property of
order 1 also have tridiagonal inverse \cite[Section 4.5]{Chen2017}.
In particular, the inverse of the DC kernel matrix is \begin{align}\label{eq:DCinv}
P_{\rm DC}^{-1}=\frac{1}{c(1-\rho^2)}
\left[
\begin{array}{ccccccc}
\frac 1{\lambda} & -\frac{\rho}{\lambda^{\frac32}} &\cdots  &0 &0\\
-\frac{\rho}{\lambda^{\frac32}} & \frac{1+\rho^2}{\lambda^{2}} &\ddots  &0 &0\\
\vdots& \ddots & \ddots & \ddots & \vdots \\
0&0  &\ddots  &\frac{1+\rho^2}{\lambda^{n-1}} & -\frac{\rho}{\lambda^{\frac{2n-1}2}} \\
0& 0 & \cdots & -\frac{\rho}{\lambda^{\frac{2n-1}2}} & \frac1{\lambda^{n}}
\end{array}
\right]
\end{align}
where $c>0$, $0<\lambda\leq 1$, and $|\rho|<1$.
Thus we have the following result on $r^*_{\rm D}$ and
$r^*_{\rm A}$ for the DC kernel.
\begin{cor}
	\label{prop6}
	Consider the DC kernel \dref{kdc}. Then for $\rho>0$ with any $N$, or for $\rho<0$ with even
	$N$,   we have $r^*_{\rm D}\neq r^\dagger~\mbox{and}~
	r^*_{\rm A}\neq r^\dagger$. In particular,
	for the TC kernel \dref{ktc}, $r^*_{\rm D}\neq r^\dagger~\mbox{and}~
	r^*_{\rm A}\neq r^\dagger$ for any $N$. 
\end{cor}


\subsubsection{More general nondiagonal kernel matrices}

Now we consider more general nondiagonal kernel matrices and we
start from a simple case when $n=2$.
\begin{prop}
	\label{prop5}
	Consider the kernel matrix
	\begin{align}
	P^{-1}= \left[
	\begin{array}{cc}
	p_1 & -e\\
	-e &p_2
	\end{array}
	\right]\label{tdk}
	\end{align}
	where $p_1>0,p_2>0$, and $e\neq 0$. Then we have
	$r^*_{\rm D}\neq r^\dagger~\mbox{and}~
	r^*_{\rm A}\neq r^\dagger$ for any $N$.
\end{prop}

Then we consider nondiagonal kernel matrices with $n\geq 2$ and $n\in\mathbb N$.

\begin{thm}
	\label{thm5}
	Suppose the kernel matrix $P$ is positive definite.
	When $N\geq 2n\!-\!2$, we have
	$$r^*_{\rm D}\neq r^\dagger$$ if at least one of the values $\{{\rm Tr} (Q(r^\dagger)^{-1}Q_i),1 \! \leq i \! \leq n-1\}$ is nonzero
	and
	\begin{align*}
	r^*_{\rm A}\neq r^\dagger
	\end{align*}
	if at least one of the values $\{{\rm Tr} (Q(r^\dagger)^{-2}Q_i),1 \! \leq i \!\leq n-1\}$ is nonzero.
\end{thm}
\begin{cor}
	\label{prop7}
	Consider the DC kernel \dref{kdc}. Then we have
	$r^*_{\rm D}\neq r^\dagger~\mbox{and}~
	r^*_{\rm A}\neq r^\dagger$, if $N\geq 2n\!-\!2$.
\end{cor}

Theorems \ref{thm6} to \ref{thm5} and Corollaries \ref{prop6} to \ref{prop7} show that there exist
cases such that  $r^*_{\rm D}$ and $r^*_{\rm A}$ are no longer
$r^\dagger$, implying that the information on the correlation of the
impulse response indeed has influence on the input design.

Finally, one may wonder whether it is possible to claim that for all nondiagonal kernels, $r^*_{\rm D}$ and $r^*_{\rm A}$
are \emph{not} equal to $r^\dagger$. Unfortunately, this claim is not true as can be seen from the following counter example.

\begin{example}
 Let $\mathscr{E}=1$, $\sigma^2=1$, and
\begin{align}
P^{-1}= \left[
\begin{array}{ccc}
1 & \frac12&-\frac18\\
\frac12 &1 &-\frac12\\
-\frac18&-\frac12&1
\end{array}
\right],\label{ce}
\end{align}
which has the eigenvalues $0.3526,0.875$, and $1.7724$, and thus is positive definite, and yields
\begin{align*}
Q(r^\dagger)^{-1}
= \left[
\begin{array}{ccc}
\frac8{15} & -\frac2{15}&0\\
-\frac2{15} &\frac{17}{30} &\frac2{15}\\
0&\frac2{15}&\frac8{15}
\end{array}
\right].
\end{align*}
Then it can be shown that ${\rm Tr} (Q(r^\dagger)^{-1}Q_1)=0$ and ${\rm Tr} (Q(r^\dagger)^{-1}Q_2)=0$, which implies that the gradient of $\log\det(\sigma^2Q(r)^{-1})$ at $r=r^\dagger$ is zero and thus $r^*_{\rm D}=r^\dagger$ for nondiagonal kernel matrix  \dref{ce} and for any $N\geq 3$.
\end{example}

\section{Numerical Simulation}
\label{sec:sim}
We illustrate by Monte Carlo simulations that the optimal input derived from the proposed method can improve the quality of the regularized FIR model estimate \dref{rls} in contrast with the white noise input.


\subsection{Test Systems}
We first use the method in \cite{Chen2012,Pillonetto2015} to generate
1000 30th order LTI systems. For each system, we truncate its impulse response at the order $50$ and obtain a FIR model of order 50 accordingly, which is treated as the test system. In this way, we generate 1000 test systems.

\subsection{Preliminary Data-Bank}

For each test system, we generate a preliminary data record as follows, whose usage is to get a preliminary estimate of the kernel matrix $P(\eta)$ and the noise variance $\sigma^2$ necessary for the input design. We simulate each test system with a white Gaussian noise input with $N=50$ and $\mathscr{E}=10$, treat the unknown inputs according to \dref{eq:assoninicond}, and get the noise-free output with $N=50$. The noise-free output is then corrupted by an
additive white Gaussian noise.
The signal-to-noise ratio  (SNR),
i.e., the ratio between the variance of the noise-free output and
the noise, is uniformly distributed over $[1,10]$. In this way, we get 1000 preliminary data records. 

\subsection{Optimal Input}

For each test system and preliminary data record, we derive the optimal input as follows: 
\begin{enumerate}
	\item[1)] We consider FIR model with order $n=50$ and  estimate the noise variance $\sigma^2$ by LS method, as described in \cite{Goodwin1992,Chen2012} and 
	the hyperparameter $\eta$ of the TC kernel matrix $P(\eta)$  given in \dref{ktc} by the EB method \dref{eq:EB}.
	
	\item[2)] Then for the obtained kernel matrix and noise variance, we derive the optimal input $u^*_{\rm D},u^*_{\rm A},$ and $u^*_{\rm E}$ by solving the input design problems \dref{dop_o}, \dref{aop_o}, \dref{eop_o}, respectively, with $N=50$ and $\mathscr{E}=10$ and the proposed two-step procedure.
For comparison, we also derive the optimal input in \cite{Fujimoto2016} with their gradient-based algorithm. It is worth to recall from (\ref{eq:FS_assoninicond}) that the unknown inputs are set to zero  in \cite{Fujimoto2016}.


\end{enumerate}

\begin{figure}[t]
	\includegraphics[scale=0.47]{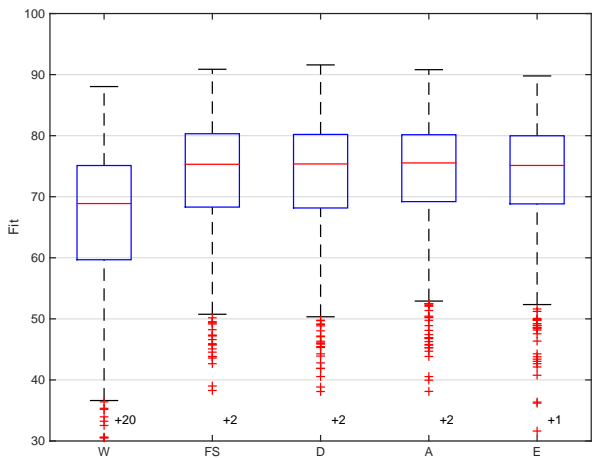}
	\caption{Boxplots of the 1000 fits of the RLS estimates based on the preliminary data-bank and test databank.}
	\label{f1}
\end{figure}

\subsection{Test Data-Bank}

For each test system and each optimal input, we generate a test data record as follows, whose usage is to illustrate the efficacy of the input design. We simulate each test system with the optimal input, treat the unknown inputs according to \dref{eq:assoninicond}, and get the noise-free output with $N=50$. The noise-free output is then corrupted by an
additive white Gaussian noise with the same variance as the additive white Gaussian noise in the preliminary data record.
In this way, for each test system and each optimal input, we get 1000 test data records.

\subsection{Simulation Results}

The simulation results are summarized in  Table \ref{tab1} and Fig. \ref{f1}, where W, FS, D, A, E are used to denote the corresponding simulation results for different inputs, respectively, and W denotes the result corresponding to the white noise input in the preliminary data-bank. In particular,
Table \ref{tab1} shows the average model fits and Fig. \ref{f1} shows the
boxplots of the 1000 model fits, where the model fit \citep{Ljung2012} is defined as follows
\begin{align*}
\mbox{\rm Fit} =100\times \left( 1 - \frac{\|\widehat{\theta}_N^{\rm R} - \theta_0\|}{\|\theta_0-\tilde{\theta}_0\|}\right),~~\tilde{\theta}_0=\frac{1}{50}\sum_{k=1}^{50}g_k^0
\end{align*}
where $\widehat{\theta}_N^{\rm R}$ represents the regularized FIR model estimate \dref{rls} corresponding to each data record, the regression matrix $\Phi$ is constructed by the input according to \dref{eq:assoninicond}, and the estimated noise variance and TC kernel matrix from the preliminary data record are used.

\subsection{Findings}
In contrast with the white noise input, all optimal inputs including the one given in \cite{Fujimoto2016} and $u^*_{\rm D},u^*_{\rm A},$ and $u^*_{\rm E}$ given in this paper improve the average fit of the regularized FIR model estimate (\ref{rls}) from 66 to 73 and are more robust.
An intuitive explanation for this improvement is that the data records generated by the optimal inputs may have higher average SNRs than the preliminary data record generated by the white noise input.
This guess is verified for the data records in this experiment.
In addition, the fits of the estimates of the data records generated by all optimal inputs are quite close.





\begin{table}[!ht]
	\centering
	\caption{Average fits of the regularized FIR model estimate \dref{rls}
		given in Fig. \ref{f1}
		based on 1000 Monte Carlo runs.}
	\vspace{1.5ex}
	\begin{tabu} to 0.48
		\textwidth{X[1,c]X[1,c]X[1,c]X[1,c]X[1,c]X[1,c]}
		\hline
		 W &  	\mbox{FS} & 		  D 	&	A	&	 E 	\\	\hline
		66.24	&	73.56	 	&	73.44 	&	73.87 	&	73.46 \\ \hline
	\end{tabu}
	\label{tab1}
\end{table}


\section{Conclusions}
\label{con}
In this paper, the input design of kernel-based regularization method for LTI system identification has been investigated by minimizing the scalar measures of the Bayesian MSE matrix. Under suitable assumption on the unknown inputs and with the introduction of a quadratic map, the non-convex optimization problem associated with the input design problem is transformed into a convex optimization problem. Then based on the global minima of the convex optimization problem, the optimal input is derived by solving the inverse image of the quadratic map. For some special types of fixed kernels, some analytic results provide insights on the input design and
also its dependence on the kernel structure.

\appendix
\textbf{Appendix A}

Appendix A contains the proof of the results in the paper, for which
the technical lemmas are placed in Appendix B.
The proofs of Propositions \ref{prop6} and \ref{prop2}
are straightforward and thus omitted.

\renewcommand{\thesection}{A}

\subsection{Proof of Theorem  \ref{thm1}}

\setcounter{thm}{0}
\renewcommand{\thethm}{A\arabic{thm}}
\setcounter{lem}{0}
\renewcommand{\thelem}{A\arabic{lem}}
\setcounter{rem}{0}
\renewcommand{\therem}{A\arabic{rem}}

First, we introduce the cyclic permutation matrix $C\in\mathbb R^N$  defined by
\begin{align}
C\eq \left[
\begin{array}{cccccc}
0&I_{N-1}\\
1&0
\end{array}
\right],
\end{align}
which has the following properties
\begin{align}
C^iC^j=C^{i+j},C^N=I_N,(C^i)^T=C^{N-i}.\label{cpmp}
\end{align}
Note that $L_0=I_N$ and $L_j=(C^j+(C^j)^T)/2$ is real symmetric for $ j=1,\dots, n-1$.
Clearly, $L_i$ commutes with $L_j$ by using the properties \dref{cpmp} for $ i,j=1,\dots, n-1$ and $i\neq j$.
It follows from Theorem \ref{thm9} in Appendix B that
there exists an orthogonal matrix $W$
such that
\begin{align}
W^T L_j W = S_j,~j=0,\dots, n-1,
\end{align}
where $S_j=\diag[s_{j0},\cdots,s_{j,N-1}]$ is a real diagonal matrix.
Therefore, we have
\begin{align}
f_j(u) = u^TWS_jW^Tu
=\sum_{i=0}^{N-1}s_{ji}z_i^2
\end{align}
where $z=[z_0,z_1,\cdots,z_{N-1}]^T=W^Tu$.
We obtain
\begin{align}
f(u)\!=\!
\underbrace{\left[
	\begin{array}{cccc}
	s_{00}&s_{01}&\cdots&s_{0,N-1}\\
	s_{10}&s_{11}&\cdots&s_{1,N-1}\\
	\vdots &\vdots&\ddots&\vdots\\
	s_{n-1,0}&s_{n-1,1}&\cdots&s_{n-1,N-1}
	\end{array}
	\right]}_{S}
\!\!
\left[
\begin{array}{c}
z_0^2\\
z_1^2\\
\vdots\\
z_{N-1}^2\\
\end{array}
\right].
\end{align}    	
Since $W^TW=I_N$,  the orthogonal linear mapping $h_3(\cdot)$ has the property that
 the image is $\mathscr{Z}=\{z|z^Tz= \mathscr{E}\}$.
Then the image of $h_2(\cdot)$ under $\mathscr{Z}$ is the polytope $\mathscr{X}$ given in \dref{ih2} by the definition of $h_2(\cdot)$. Finally, the conclusion that $\mathscr{F}$ is a polytope follows from that a linear transform of a polytope is still a polytope \cite[p. 18, 2.5]{Brondsted1983}. This completes the proof.

\subsection{Proof of Theorem \ref{thm2a}}

We only prove the case for even $N$. The proof for odd $N$ is similar and thus omitted.

For convenince, we place the construction of the matrices $W$ and $S$ in Lemma \ref{lemb4}. So we first consider the claim $\rank(S)=\min(N/2+1,n)$, which follows from that the vectors $\xi_0,\xi_1,\cdots,\xi_{\frac{N}{2}}$ are orthogonal to each other and
$\xi_j=\xi_{N-j}$ for $ j=0,\dots, N$.

Then note that the polyhedron $\mathscr{X}$ in Theorem \ref{thm1}
can be written as a polytope
\begin{align}
\mathscr{X}=\mathbf{conv}\left\{d_0,d_1,\cdots,d_{N-1}\right\}\label{simplex},
\end{align}
where for $ j=0,\dots, N-1$, $d_j$ is the vector whose elements are zero except the $j\!+\!1$-element that is equal to $\mathscr{E}$. Since $\cos(lj\varpi) = \cos((N-l)j\varpi)$  for any integer $l$ and $ j=0,\dots, N\!-\!1$,
the matrix $S$ can be expressed in the form of \eqref{seven}.

Now we denote by $\mathscr{C}$ the convex hull of the column vectors of $S$ in \eqref{seven} and we will show that $\mathscr{F}=\mathscr{C}$. On the one hand, note that
\begin{align}
Sd_j = Sd_{N-j}=\mathscr{E}\xi_{j}(1\!:\!n),~j=0,\dots, N/2.\label{sd}
\end{align}
This implies that the image of \dref{simplex} under the linear transform $h_1(x)=Sx$ is included in $\mathscr{C}$
and hence $\mathscr{F}\subset \mathscr{C}$.
On the other hand, for each element $r$ in $\mathscr{C}$,
there exist
$a_j\geq 0, j=0,\dots, N/2$ with $\sum_{j=0}^{N/2}a_j=1$ such that
$$r=\sum_{j=0}^{N/2}a_j\mathscr{E}\xi_j(1\!:\!n)
=S\Big(\sum_{j=0}^{N/2}a_jd_j\Big)\in \mathscr{F}$$
where \dref{sd} is used.
This concludes $\mathscr{C}\subset\mathscr{F}$.
Therefore, we have $\mathscr{F}=\mathscr{C}$. This completes the proof.

\subsection{Proof of Proposition \ref{thm2}}
First, we show that $r^\dagger\in\mathscr{F}$.
This can be done by finding a point $x^0\in\mathscr{X}$ such that $Sx^0=r^\dagger$.
We can choose $x^0=\mathscr{E}[1,\cdots,1]^T/N\in\mathscr{X}$.
Applying the property that $\xi_0$ is orthogonal to the vectors $\{\xi_j,1\leq j\leq N/2\}$ for even $N$ or  $\{\xi_j,1\leq j\leq (N\!-\!1)/2\}$ for odd $N$ yields $Sx^0=r^\dagger$.
Then we check the gradient of $\log\det(\sigma^2Q(r)^{-1})$
with respective to the last $n\!-\!1$ variables of $r$ is zero at $r=r^\dagger$.
For ridge kernels, we have $\sigma^2P^{-1}= \sigma^2I_n/c $
and accordingly $Q(r^\dagger)$ is $(\mathscr{E}+\sigma^2/c)I_n$.
Thus we find that
\begin{align*}
 -{\rm Tr} (Q(r^\dagger)^{-1}Q_i)=0,~ i=1,\dots, n\!-\!1.
\end{align*}
Since the problem \dref{odr} is strictly convex, $r^\dagger$ is the unique stationary point and $r^*_{\rm D} =r^\dagger$.

The proof for $r^*_{\rm A}=r^\dagger$ can be derived in a similar way. 

\subsection{Proof of Theorem \ref{thm3}}
The proof of Theorem \ref{thm3} is the same as that of Proposition \ref{thm2} since  $Q(r^\dagger)$ is still diagonal in this case.

\subsection{Proof of Theorem \ref{thm6}}
First, we consider the case $e_j>0$ for $j=2, \dots, n$.
In this case, we have each element of $Q(r^\dagger)^{-1}=(\mathscr{E}I_n +\sigma^2P^{-1})^{-1}$ is positive by Theorem \ref{thma3}.
This implies that each element of the gradient at  $r=r^\dagger$ is negative by \dref{gvd} and the definition of $Q_i$. It follows that
$$[\nabla \log\det(\sigma^2Q(r^\dagger)^{-1})]^T
{\xi}_0(2:n)
< 0,$$
since all elements of ${\xi}_0(2:n)$ are equal to one.
This violates Lemma \ref{lemb2} in Appendix B, which gives the necessary and sufficient condition for $r^*_{\rm D}=r^\dagger$.
Thus $r^*_{\rm D}\neq r^\dagger$ for this case.

Then we consider  the case where $e_j<0$ for $j=2, \dots, n$ and $N$ is even.
In this case, we have $(Q(r^\dagger)^{-1})_{ij}>0$ when $|i-j|$ is even and $(Q(r^\dagger)^{-1})_{ij}<0$ when $|i-j|$ is odd by Theorem \ref{thma3}.
This yields that the sign of the $i$-th element
$\frac{\partial \log\det(\sigma^2 Q(r^\dagger)^{-1})}{\partial r_i}$ of the gradient  $\nabla \log\det(\sigma^2Q(r^\dagger)^{-1})$ is positive for odd $i$ and negative for even $i$.
This shows that
$$[\nabla \log\det(\sigma^2Q(r^\dagger)^{-1})]^T
{\xi}_{N/2}(2:n)
< 0$$
since the $i$-th element of ${\xi}_{N/2}(2:n)$ is $(-1)^{i}$.
This violates  the necessary and sufficient condition for $r^*_{\rm D}=r^\dagger$, and thus $r^*_{\rm D}\neq r^\dagger$ for this case.

The proof of $r^*_{\rm A}\neq r^\dagger$ is similar and thus omitted.

\subsection{Proof of Corollary \ref{prop6}}
The proof is trivial by noting (\ref{eq:DCinv}) and Theorem \ref{thm6}.


\subsection{Proof of Proposition \ref{prop5}}
It follows that
\begin{align*}
Q(r^\dagger)^{-1}\!=\!
\frac{1}{\det(Q(r^\dagger))}
\left[
\begin{array}{ccccc}
\mathscr{E}\!+\!\sigma^2p_2&\sigma^2e\\
\sigma^2e&\mathscr{E}\!+\!\sigma^2p_1
\end{array}
\right].
\end{align*}
Then we have
$-{\rm Tr} (Q(r^\dagger)^{-1}Q_1)=-2\sigma^2e/\det(Q(r^\dagger))$.
In this case, we have ${\xi}_j(2:n)={\xi}_j(2:2)=\cos(j\varpi)$.
As a result, there always exists an index $j$ such that $e$ and $\cos(j\varpi)$ have the same sign for  $j=0,\dots, N/2 $ when $N$ is even and for $j=0,\dots, (N\!-\!1)/2$ when $N$ is odd.
One obtains that $-e \cos(j\varpi)<0$ for this $j$,
which violates \dref{b19}.
Therefore,  $r^*_{\rm D}\neq r^\dagger$.
The proof of $r^*_{\rm A}\neq r^\dagger$ is similar and is omitted.

\subsection{Proof of Theorem \ref{thm5}}
It follows from Lemma \ref{lemb1} in Appendix B that
$0$ is an interior point of  $\mathscr{F}'=\{r'\in\mathbb{R}^{n-1}|[\mathscr{E}, r'^T]^T\in\mathscr{F}\}$
 when $N\geq 2n-2$.
Since $P^{-1}$ is nondiagonal, we have $Q(r^\dagger)^{-1}$ is also nondiagonal.
The condition that at least one of the values ${\rm Tr} (Q(r^\dagger)^{-1}Q_i)$, $ i=1, \dots, n-1$, is nonzero implies that
 $r^\dagger$ is not a stationary point of the problem \dref{odr}.
This means that $r^*_{\rm D}\neq r^\dagger$. The proof for $r^*_{\rm A}\neq r^\dagger$ is similar and thus omitted.

\subsection{Proof of Corollary \ref{prop7}}
For the DC kernel \dref{kdc}, by Theorem \ref{thma3} as used in the proof of Theorem \ref{thm6} we have
\begin{enumerate}[1)]
	\item When $\rho>0$, $(Q(r^\dagger)^{-1})_{ij}>0$ and $(Q(r^\dagger)^{-2})_{ij}>0$.
	\item When $\rho<0$, $(Q(r^\dagger)^{-1})_{ij}>0$ and $(Q(r^\dagger)^{-2})_{ij}>0$ if $|i-j|$ is even and $(Q(r^\dagger)^{-1})_{ij}<0$ and $(Q(r^\dagger)^{-2})_{ij}<0$ if $|i-j|$ is odd.
\end{enumerate}
This implies that ${\rm Tr} (Q(r^\dagger)^{-1}Q_i)\neq 0$ and ${\rm Tr} (Q(r^\dagger)^{-2}Q_i)\neq 0$ for all indices $i=1, \dots, n-1$.
As a result, the proposition is true by Theorem \ref{thm5}.

\section*{Appendix B}

\renewcommand{\thesection}{B}

\setcounter{subsection}{0}

\setcounter{thm}{0}
\renewcommand{\thethm}{B\arabic{thm}}
\setcounter{lem}{0}
\renewcommand{\thelem}{B\arabic{lem}}
\setcounter{rem}{0}
\renewcommand{\therem}{B\arabic{rem}}

This appendix contains the technical lemmas used in the proof in
Appendix A.
\begin{thm}
	\label{thm9}
	\citep[Theorem 9]{Jiang2016}
	Real symmetric matrices $\{I_N,A_1,\cdots,A_m\}$ with $A_i\in\mathbb R^{N\times N}$ and $i=1,\dots,m$ are simultaneously diagonalizable via an orthogonal congruent matrix if and only if $A_i$ commutes with $A_j$ for $i,j=1,2,\cdots,m$ and $i\neq j$.	
\end{thm}

\begin{thm}
	\label{thma4}
    (\citet[Theorem 3.1]{Gray2006};\citet{Tee2007})
	Denote the circulant matrix $B$ generated by a row vector $b=[b_0,b_1,\cdots,b_{N-1}]$
	by
	\begin{align}
	B={\rm circ}(b)
	\eq \left[
	\begin{array}{ccccc}
	b_0&b_1&\ddots&b_{N-2}&b_{N-1}\\
	b_{N-1}&b_0&\ddots&b_{N-3}&b_{N-2}\\
	\ddots&\ddots&\ddots&\ddots&\ddots\\
	b_2&b_3&\ddots&b_0&b_1\\
	b_1&b_2&\ddots&b_{N-1}&b_0
	\end{array}
	\right].
	\end{align}	
	Then $B$ has unit eigenvectors
	\begin{align*}v^{(m)}&=\frac1{\sqrt{N}}\big[1,\exp(-\iota\varpi m),\cdots,\exp(-\iota\varpi (N\!-\!1) m)\big]^T, \\ m&=0,\cdots,N\!-1,\end{align*}
	where $\varpi=2\pi/N$ and $\iota$ is the imaginary unit ($\iota^2=-1$),
	and the corresponding eigenvalues
	\begin{align*}
	\tau^{(m)}&=\sum_{k=0}^{N-1}b_k\exp(-\iota mk\varpi)\\
	&=b[1,\exp(-\iota m\varpi),\cdots,\exp(-\iota m(N-1)\varpi)]^T
	\end{align*}
	and can be expressed by
	\begin{align*}
	B=A~\!\diag([\tau^{(0)},\cdots,\tau^{(N-1)}])A^{H}
	\end{align*}
	where $A=[v^{(0)},\cdots,v^{(N-1)}]$ is unitary and $A^H$ denotes the complex conjugate transpose of $A$. 
\end{thm}

\begin{thm}
	\label{thma3}
	Let $A$ be a symmetric and positive definite tridiagonal matrix of dimension $n$,
	\begin{align}
	A=\left[
	\begin{array}{ccccc}
	a_1&-b_2&&&\\
	-b_2&a_2&-b_3&&\\
	&\ddots&\ddots&\ddots&\\
	&&-b_{n-1}&a_{n-1}&-b_n\\
	&&&-b_n&a_n
	\end{array}
	\right]
	\end{align}
	and denote the $(i,j)$-element of $A^{-1}$ by $(A^{-1})_{ij}$.
	Then we have
	$(A^{-1})_{ij}>0$ when $b_l> 0$  for $ l=2, \dots, n$,
	while $(A^{-1})_{ij}<0$ if $|i-j|$ is odd and
	$(A^{-1})_{ij}>0$ if $|i-j|$ is even
	when $b_l< 0$  for $ l=2, \dots, n$.
	
\end{thm}
\begin{proof}
	The result is obtained by applying Theorem 2.3 of \cite{Meurant1992} when $A$ is positive definite.
\end{proof}

\begin{thm}
	\label{oct}
	\citep[Section 4.2.3, page 139]{Boyd2004}
	Let $\varphi:\mathbb{R}^m\!\xra{}\!\mathbb{R}$ be a differentiable convex function, and let $\mathscr{B}\subset\mathbb{R}^m$ be a nonempty closed convex set.
	Consider the problem
	\begin{align}
	minimize~\varphi(h)~~subject~ to ~h\in \mathscr{B}.
	\end{align}
	A vector $h^*$ is optimal for this problem if and only if $h^* \in \mathscr{B}$ and
	\begin{align}
	\nabla \varphi(h^*)^T(y-h^*)\geq 0~~\mbox{for all}~~y\in \mathscr{B}
	\end{align}	
	where $$\nabla \varphi(\cdot)\eq\left[\frac{\partial \varphi(\cdot)}{\partial h_1},\cdots,\frac{\partial \varphi(\cdot)}{\partial h_{m}}\right]^T$$ is the gradient of $\varphi(\cdot)$ with respect to $h=[h_1,\cdots,h_m]^T$.
\end{thm}

\begin{lem} \label{lemb4}
	The matrices $\{L_j,j=0,\cdots,n-1\}$ can be diagonalized simultaneously by the matrix $W$, i.e.,
	\begin{align}
	W^T L_j W = S_j=\diag(\xi_j),~ j=0,\dots, n-1.
	\end{align}
\end{lem}
\begin{proof}
	First, we  consider the case where $N$ is even.
	It follows that $L_0=I_N={\rm circ}([1,0,\cdots,0])$ and $L_j={\rm circ}(l_j/2),~j=1,\cdots,n-1$, where $l_j$ is the $N$-dimensional row vector whose elements are zero except that the $j\!+\!1$-th and $N\!-\!j+\!1$-th elements are one. Then by Theorem \ref{thma4}, the eigenvalues of $L_j$ are
	\begin{align*}
	\tau^{(m,j)}&=l_j[1,\exp(-\iota m\varpi),\cdots,\exp(-\iota m(N-1)\varpi)]^T/2\\
	&=\exp(-\iota mj\varpi)/2 + \exp(-\iota m(N-j)\varpi)/2\\
	&=\cos(mj\varpi), \ m\!=\!0,\cdots,N\!-\!1,
	\end{align*}
 the eigenvectors $\{v^{(m)},m=0,\cdots,N-1\}$ is an orthonormal basis for all matrices $L_j,j=0,\cdots,n-1$. Moreover,
 we have $v^{(0)}=\xi_0/\sqrt{N}$ and
 $v^{(N/2)}=\xi_{\frac N2}/\sqrt{N}$ are real, and $\tau^{(m,j)}\!=\!\tau^{(N-m,j)}$, which implies that the eigenvectors $v^{(m)}$  and $v^{(N-m)}$ of $L_j$ correspond to the same eigenvalue $\cos(mj\varpi)$.
Thus the linear combinations $(v^{(m)}+{v}^{(N-m)})/2=\xi_m/\sqrt{N}$ and $\iota({v}^{(N-m)}-v^{(m)})/2=\zeta_m/\sqrt{N}$ of the eigenvectors $v^{(m)}$ and $v^{(N-m)}$ corresponding to the $m$-th and $(N\!-\!m)$-th eigenvalue $\cos(mj\varpi)$ of $L_j$ are real and orthogonal. 
Further we have
	 $|\xi_m|=|\zeta_m|=\sqrt{N/2}$,
	and thus $$\Big\{\xi_0,\sqrt{2}\xi_1,\cdots,\sqrt{2}\xi_{\frac{N-2}{2}},\xi_{\frac{N}{2}},\sqrt{2}\zeta_{\frac{N-2}{2}},\cdots,\sqrt{2}\zeta_1\Big\}\Big/\sqrt{N}$$
	is a group of real orthonormal eigenvectors of $L_j$ and the corresponding eigenvalues are $\{1,\cos(j\varpi),\cdots,\cos((N\!-\!1)j\varpi)\}$, which are actually the elements of $\xi_j$.
    This completes the proof for the case for even $N$. The proof for odd $N$ is similar and thus omitted.
\end{proof}

\begin{lem} \label{lemb2}
	The vector $r^\dagger\in \mathscr{F}$ is the solution of the problem
	\begin{align}
	minimize~\varphi(r)~~subject~ to ~r\in \mathscr{F} \label{cf}
	\end{align}
	where $\varphi(r)$ represents $\log\det( \sigma^2Q(r)^{-1})$ or ${\rm Tr}(\sigma^2 Q(r)^{-1})$,
	if and only if
	\begin{align}
	[\nabla \varphi(r^\dagger)]^T
	{\xi}_j(2:n)
	\geq 0\label{b19}
	\end{align}
	with $j=0, \dots, N/2 $ for even $N$ or $j=0, \dots, (N\!-\!1)/2$ for odd $N$,
	where $\nabla \varphi(r^\dagger)=\left[\frac{\partial \varphi(r^\dagger)}{\partial r_1},\cdots,\frac{\partial \varphi(r^\dagger)}{\partial r_{n-1}}\right]^T$.
\end{lem}	
\begin{proof}
Note that the first element of each member in $\mathscr{F}$ is $\mathscr{E}$ and hence the first element of the difference $r - r^\dagger$ is zero for any $r\in \mathscr{F}$.
In addition, we have the last $n-1$ elements of $r^\dagger$ are zero.
By applying Theorem \ref{oct}, we see that $r^\dagger$ is the solution of \dref{cf} if and only if
\begin{align}
[\nabla \varphi(r^\dagger)]^T (\tilde{r} - 0)\geq 0,\label{oc}
\end{align}
where $\tilde{r}$ is any convex linear combination of the points $\{{\xi}_j(2:n),j=0,\cdots,N/2\}$ when $N$ is even, namely,
$\tilde{r}=\sum_{j=0}^{N/2}a_j {\xi}_j(2:n)$ with $a_j\geq0$ and $\sum_{j=0}^{N/2}a_j=1$.
Clearly, the condition \dref{oc} is equivalent to \eqref{b19} for $j=0,\cdots,N/2$.
The case when $N$ is odd can be proved in a similar way and thus omitted.
\end{proof}

\begin{lem} \label{lemb1}
	The zero vector $0\in \mathbb{R}^{n-1}$ is an interior point of  $\mathscr{F}'=\{r'\in\mathbb{R}^{n-1}|[\mathscr{E}, r'^T]^T\in\mathscr{F}\}$ if $N\geq 2n\!-\!2$.
\end{lem}
\begin{proof}
		This conclusion is proved by considering two cases, respectively: 1) for even $N$ with $N\geq 2n\!-\!2$; 2) for odd $N$ with $N\geq 2n\!-\!1$.
		Let us explore the first case.
		It is clear that the dimension of $\mathscr{F}'$ is equal to or less than $n-1$.
		It follows from Theorem \ref{thm2a} that $\rank(S)=n$.
		This means that
		the dimension of $\mathscr{F}'$ is $n-1$
		and $\mathscr{F}'$ consists of all convex combinations of affinely independent  $n$ vectors from $\{\mathscr{E}\xi_j(2\!:\!n),j=0,\dots, N/2\}$ \citep[p. 14, Corollary 2.4]{Brondsted1983}.
	Now we will prove the result by contradiction.
	Assume that $0$ is not an interior point of $\mathscr{F}'$.
	Thus $0$ must be located on a facet of $\mathscr{F}'$ and the dimension of this facet is $n-2$ since the dimension of $\mathscr{F}'$ is $n-1$ and hence $0$ is a convex combination of at most $n-1$  affinely independent vectors from  $\{\mathscr{E}\xi_j(2\!:\!n),j=0,\dots, N/2\}$ \citep[Corollary 2.4]{Brondsted1983}.
	In contrast, $0$ can be expressed by a convex combination in the following way
	\begingroup
	\allowdisplaybreaks
	\begin{align}
	\nonumber
	0=\frac{1}{N}(\mathscr{E}S(2:n,:)) \xi_0
	&=\frac{1}{N}\big(\mathscr{E}\xi_0(2\!:\!n) + \mathscr{E}\xi_{\frac{N}{2}}(2\!:\!n)\big) \\
	&~~~~+ \frac{2}{N} \sum_{j=1}^{\frac N2-1}(\mathscr{E}\xi_j(2\!:\!n))\label{expr}
	\end{align}
	\endgroup
	due to the orthogonal vectors $\{\xi_j,j=0,1,\cdots\}$ and the expression \dref{seven} of $S$.
	Further, the expression \dref{expr} can  be rewritten as a convex combination of
	affinely independent $n$ vectors from $\{\xi_j(2\!:\!n),j=0,\dots, N/2\}$ and all the weights are positive since the remaining $N/2+1-n$ vectors of $\{\xi_j(2\!:\!n),j=0,\dots, N/2\}$ are convex combinations of these $n$ affinely independent vectors.
	This contradiction means that the proposition is true.
	The proof for odd $N$ with $N\geq 2n\!-\!1$ is similar and thus omitted.
\end{proof}


\begin{thebibliography}{35}
	\expandafter\ifx\csname natexlab\endcsname\relax\def\natexlab#1{#1}\fi
	\providecommand{\url}[1]{\texttt{#1}}
	\providecommand{\href}[2]{#2}
	\providecommand{\path}[1]{#1}
	\providecommand{\DOIprefix}{doi:}
	\providecommand{\ArXivprefix}{arXiv:}
	\providecommand{\URLprefix}{URL: }
	\providecommand{\Pubmedprefix}{pmid:}
	\providecommand{\doi}[1]{\href{http://dx.doi.org/#1}{\path{#1}}}
	\providecommand{\Pubmed}[1]{\href{pmid:#1}{\path{#1}}}
	\providecommand{\bibinfo}[2]{#2}
	\ifx\xfnm\relax \def\xfnm[#1]{\unskip,\space#1}\fi
	\bibitem[{Boyd \& Vandenberghe(2004)}]{Boyd2004}
	\bibinfo{author}{Boyd, S.}, \& \bibinfo{author}{Vandenberghe, L.}
	(\bibinfo{year}{2004}).
	\newblock {\it \bibinfo{title}{Convex optimization}\/}.
	\newblock \bibinfo{publisher}{Cambridge university press}.
	\bibitem[{Br\o{}ndsted(1983)}]{Brondsted1983}
	\bibinfo{author}{Br\o{}ndsted, A.} (\bibinfo{year}{1983}).
	\newblock {\it \bibinfo{title}{An introduction to convex polytopes}\/}.
	\newblock \bibinfo{address}{New York, Inc}:
	\bibinfo{publisher}{Springer-Verlag}.
	\bibitem[{Carli et~al.(2017)Carli, Chen \& Ljung}]{Carli2017}
	\bibinfo{author}{Carli, F.~P.}, \bibinfo{author}{Chen, T.}, \&
	\bibinfo{author}{Ljung, L.} (\bibinfo{year}{2017}).
	\newblock \bibinfo{title}{Maximum entropy kernels for system identification}.
	\newblock {\it \bibinfo{journal}{IEEE Transactions on Automatic Control}\/},
	{\it \bibinfo{volume}{62}\/}, \bibinfo{pages}{1471--1477}.
	\bibitem[{Chen et~al.(2014)Chen, Andersen, Ljung, Chiuso \&
		Pillonetto}]{Chen2014p}
	\bibinfo{author}{Chen, T.}, \bibinfo{author}{Andersen, M.~S.},
	\bibinfo{author}{Ljung, L.}, \bibinfo{author}{Chiuso, A.}, \&
	\bibinfo{author}{Pillonetto, G.} (\bibinfo{year}{2014}).
	\newblock \bibinfo{title}{System identification via sparse multiple
		kernel-based regularization using sequential convex optimization techniques}.
	\newblock {\it \bibinfo{journal}{IEEE Transactions on Automatic Control}\/},
	{\it \bibinfo{volume}{59}\/}, \bibinfo{pages}{2933--2945}.
	\bibitem[{Chen et~al.(2016)Chen, Ardeshiri, Carli, Chiuso, Ljung \&
		Pillonetto}]{Chen2016}
	\bibinfo{author}{Chen, T.}, \bibinfo{author}{Ardeshiri, T.},
	\bibinfo{author}{Carli, F.~P.}, \bibinfo{author}{Chiuso, A.},
	\bibinfo{author}{Ljung, L.}, \& \bibinfo{author}{Pillonetto, G.}
	(\bibinfo{year}{2016}).
	\newblock \bibinfo{title}{Maximum entropy properties of discrete-time
		first-order stable spline kernel}.
	\newblock {\it \bibinfo{journal}{Automatica}\/},  {\it \bibinfo{volume}{66}\/},
	\bibinfo{pages}{34--38}.
	\bibitem[{Chen \& Ljung(2017)}]{Chen2017}
	\bibinfo{author}{Chen, T.}, \& \bibinfo{author}{Ljung, L.}
	(\bibinfo{year}{2017}).
	\newblock \bibinfo{title}{On kernel design for regularized lti system
		identification}.
	\newblock {\it \bibinfo{journal}{arXiv preprint arXiv:1612.03542}\/}, .
	\bibitem[{Chen et~al.(2012)Chen, Ohlsson \& Ljung}]{Chen2012}
	\bibinfo{author}{Chen, T.}, \bibinfo{author}{Ohlsson, H.}, \&
	\bibinfo{author}{Ljung, L.} (\bibinfo{year}{2012}).
	\newblock \bibinfo{title}{On the estimation of transfer functions,
		regularizations and gaussian processes---revisited}.
	\newblock {\it \bibinfo{journal}{Automatica}\/},  {\it \bibinfo{volume}{48}\/},
	\bibinfo{pages}{1525--1535}.
	\bibitem[{Chiuso(2016)}]{Chiuso2016}
	\bibinfo{author}{Chiuso, A.} (\bibinfo{year}{2016}).
	\newblock \bibinfo{title}{Regularization and bayesian learning in dynamical
		systems: Past, present and future}.
	\newblock {\it \bibinfo{journal}{Annual Reviews in Control}\/},  {\it
		\bibinfo{volume}{41}\/}, \bibinfo{pages}{24--38}.
	\bibitem[{Fujimoto \& Sugie(2016)}]{Fujimoto2016}
	\bibinfo{author}{Fujimoto, Y.}, \& \bibinfo{author}{Sugie, T.}
	(\bibinfo{year}{2016}).
	\newblock \bibinfo{title}{Informative input design for kernel-based system
		identification}.
	\newblock In {\it \bibinfo{booktitle}{Proceedings of IEEE Conference on
			Decision and Control}\/} (pp. \bibinfo{pages}{4636--4639}).
	\newblock \bibinfo{organization}{IEEE}.
	\bibitem[{Gevers(2005)}]{Gevers2005}
	\bibinfo{author}{Gevers, M.} (\bibinfo{year}{2005}).
	\newblock \bibinfo{title}{Identification for control: From the early
		achievements to the revival of experiment design}.
	\newblock {\it \bibinfo{journal}{European journal of control}\/},  {\it
		\bibinfo{volume}{11}\/}, \bibinfo{pages}{335--352}.
	\bibitem[{Goodwin et~al.(1992)Goodwin, Gevers \& Ninness}]{Goodwin1992}
	\bibinfo{author}{Goodwin, G.~C.}, \bibinfo{author}{Gevers, M.}, \&
	\bibinfo{author}{Ninness, B.} (\bibinfo{year}{1992}).
	\newblock \bibinfo{title}{Quantifying the error in estimated transfer functions
		with application to model order selection}.
	\newblock {\it \bibinfo{journal}{IEEE Transactions on Automatic Control}\/},
	{\it \bibinfo{volume}{37}\/}, \bibinfo{pages}{913--928}.
	\bibitem[{Goodwin \& Payne(1977)}]{Goodwin1977}
	\bibinfo{author}{Goodwin, G.~C.}, \& \bibinfo{author}{Payne, R.~L.}
	(\bibinfo{year}{1977}).
	\newblock {\it \bibinfo{title}{Dynamic system identification: experiment design
			and data analysis}\/}.
	\newblock \bibinfo{address}{New York}: \bibinfo{publisher}{Academic press}.
	\bibitem[{Grant \& Boyd(2016)}]{Grant2016}
	\bibinfo{author}{Grant, M.~C.}, \& \bibinfo{author}{Boyd, S.~P.}
	(\bibinfo{year}{2016}).
	\newblock \bibinfo{title}{Cvx: Matlab software for disciplined convex
		programming, version 2.1}.
	\newblock \bibinfo{howpublished}{Available from http://cvxr.com/cvx/}.
	\bibitem[{Gray(2006)}]{Gray2006}
	\bibinfo{author}{Gray, R.~M.} (\bibinfo{year}{2006}).
	\newblock {\it \bibinfo{title}{Toeplitz and circulant matrices: A review}\/}.
	\newblock \bibinfo{publisher}{Now Publishers, Inc.}
	\bibitem[{Hildebrand \& Gevers(2003)}]{Hildebrand2003}
	\bibinfo{author}{Hildebrand, R.}, \& \bibinfo{author}{Gevers, M.}
	(\bibinfo{year}{2003}).
	\newblock \bibinfo{title}{Identification for control: optimal input design with
		respect to a worst-case $\nu$-gap cost function}.
	\newblock {\it \bibinfo{journal}{SIAM Journal on Control and Optimization}\/},
	{\it \bibinfo{volume}{41}\/}, \bibinfo{pages}{1586--1608}.
	\bibitem[{Hjalmarsson(2005)}]{Hjalmarsson2005}
	\bibinfo{author}{Hjalmarsson, H.} (\bibinfo{year}{2005}).
	\newblock \bibinfo{title}{From experiment design to closed-loop control}.
	\newblock {\it \bibinfo{journal}{Automatica}\/},  {\it \bibinfo{volume}{41}\/},
	\bibinfo{pages}{393--438}.
	\bibitem[{Hjalmarsson(2009)}]{Hjalmarsson2009}
	\bibinfo{author}{Hjalmarsson, H.} (\bibinfo{year}{2009}).
	\newblock \bibinfo{title}{System identification of complex and structured
		systems}.
	\newblock {\it \bibinfo{journal}{European journal of control}\/},  {\it
		\bibinfo{volume}{15}\/}, \bibinfo{pages}{275--310}.
	\bibitem[{Jansson \& Hjalmarsson(2005)}]{Jansson2005}
	\bibinfo{author}{Jansson, H.}, \& \bibinfo{author}{Hjalmarsson, H.}
	(\bibinfo{year}{2005}).
	\newblock \bibinfo{title}{Input design via lmis admitting frequency-wise model
		specifications in confidence regions}.
	\newblock {\it \bibinfo{journal}{IEEE transactions on Automatic Control}\/},
	{\it \bibinfo{volume}{50}\/}, \bibinfo{pages}{1534--1549}.
	\bibitem[{Jiang \& Li(2016)}]{Jiang2016}
	\bibinfo{author}{Jiang, R.}, \& \bibinfo{author}{Li, D.}
	(\bibinfo{year}{2016}).
	\newblock \bibinfo{title}{Simultaneous diagonalization of matrices and its
		applications in quadratically constrained quadratic programming}.
	\newblock {\it \bibinfo{journal}{SIAM Journal on Optimization}\/},  {\it
		\bibinfo{volume}{26}\/}, \bibinfo{pages}{1649--1668}.
	\bibitem[{Ljung(1999)}]{Ljung1999}
	\bibinfo{author}{Ljung, L.} (\bibinfo{year}{1999}).
	\newblock {\it \bibinfo{title}{System identification: Theory for the user}\/}.
	\newblock \bibinfo{address}{Upper Saddle River, NJ}:
	\bibinfo{publisher}{Prentice-Hall}.
	\bibitem[{Ljung(2012)}]{Ljung2012}
	\bibinfo{author}{Ljung, L.} (\bibinfo{year}{2012}).
	\newblock {\it \bibinfo{title}{System Identification Toolbox for Use with
			MATLAB}\/}.
	\newblock (\bibinfo{edition}{8th} ed.).
	\newblock \bibinfo{address}{Natick, MA}: \bibinfo{publisher}{The MathWorks,
		Inc.}
	\bibitem[{Ljung et~al.(2015)Ljung, Singh \& Chen}]{Ljung2015}
	\bibinfo{author}{Ljung, L.}, \bibinfo{author}{Singh, R.}, \&
	\bibinfo{author}{Chen, T.} (\bibinfo{year}{2015}).
	\newblock \bibinfo{title}{Regularization features in the system identification
		toolbox}.
	\newblock In {\it \bibinfo{booktitle}{Proceedings of the IFAC Symposium on
			System Identification}\/} (pp. \bibinfo{pages}{745--750}).
	\newblock \bibinfo{address}{Beijing, China}.
	\bibitem[{Marconato et~al.(2016)Marconato, Schoukens \&
		Schoukens}]{Marconato2016}
	\bibinfo{author}{Marconato, A.}, \bibinfo{author}{Schoukens, M.}, \&
	\bibinfo{author}{Schoukens, J.} (\bibinfo{year}{2016}).
	\newblock \bibinfo{title}{Filter-based regularisation for impulse response
		modelling}.
	\newblock {\it \bibinfo{journal}{IET Control Theory \& Applications}\/},  {\it
		\bibinfo{volume}{11}\/}, \bibinfo{pages}{194--204}.
	\bibitem[{Mehra(1974)}]{Mehra1974}
	\bibinfo{author}{Mehra, R.} (\bibinfo{year}{1974}).
	\newblock \bibinfo{title}{Optimal input signals for parameter estimation in
		dynamic systems--survey and new results}.
	\newblock {\it \bibinfo{journal}{IEEE Transactions on Automatic Control}\/},
	{\it \bibinfo{volume}{19}\/}, \bibinfo{pages}{753--768}.
	\bibitem[{Meurant(1992)}]{Meurant1992}
	\bibinfo{author}{Meurant, G.} (\bibinfo{year}{1992}).
	\newblock \bibinfo{title}{A review on the inverse of symmetric tridiagonal and
		block tridiagonal matrices}.
	\newblock {\it \bibinfo{journal}{SIAM Journal on Matrix Analysis and
			Applications}\/},  {\it \bibinfo{volume}{13}\/}, \bibinfo{pages}{707--728}.
	\bibitem[{Mu et~al.(2017)Mu, Chen \& Ljung}]{Mu2017a}
	\bibinfo{author}{Mu, B.}, \bibinfo{author}{Chen, T.}, \&
	\bibinfo{author}{Ljung, L.} (\bibinfo{year}{2017}).
	\newblock \bibinfo{title}{On asymptotic properties of hyperparameter estimators
		for kernel-based regularization methods}.
	\newblock {\it \bibinfo{journal}{arXiv preprint arXiv:1707.00407}\/}, .
	\bibitem[{Pillonetto et~al.(2016)Pillonetto, Chen, Chiuso, De~Nicolao \&
		Ljung}]{Pillonetto2016}
	\bibinfo{author}{Pillonetto, G.}, \bibinfo{author}{Chen, T.},
	\bibinfo{author}{Chiuso, A.}, \bibinfo{author}{De~Nicolao, G.}, \&
	\bibinfo{author}{Ljung, L.} (\bibinfo{year}{2016}).
	\newblock \bibinfo{title}{Regularized linear system identification using
		atomic, nuclear and kernel-based norms: The role of the stability
		constraint}.
	\newblock {\it \bibinfo{journal}{Automatica}\/},  {\it \bibinfo{volume}{69}\/},
	\bibinfo{pages}{137--149}.
	\bibitem[{Pillonetto \& Chiuso(2015)}]{Pillonetto2015}
	\bibinfo{author}{Pillonetto, G.}, \& \bibinfo{author}{Chiuso, A.}
	(\bibinfo{year}{2015}).
	\newblock \bibinfo{title}{Tuning complexity in regularized kernel-based
		regression and linear system identification: The robustness of the marginal
		likelihood estimator}.
	\newblock {\it \bibinfo{journal}{Automatica}\/},  {\it \bibinfo{volume}{58}\/},
	\bibinfo{pages}{106--117}.
	\bibitem[{Pillonetto et~al.(2011)Pillonetto, Chiuso \&
		De~Nicolao}]{Pillonetto2011}
	\bibinfo{author}{Pillonetto, G.}, \bibinfo{author}{Chiuso, A.}, \&
	\bibinfo{author}{De~Nicolao, G.} (\bibinfo{year}{2011}).
	\newblock \bibinfo{title}{Prediction error identification of linear systems: a
		nonparametric gaussian regression approach}.
	\newblock {\it \bibinfo{journal}{Automatica}\/},  {\it \bibinfo{volume}{47}\/},
	\bibinfo{pages}{291--305}.
	\bibitem[{Pillonetto \& De~Nicolao(2010)}]{Pillonetto2010}
	\bibinfo{author}{Pillonetto, G.}, \& \bibinfo{author}{De~Nicolao, G.}
	(\bibinfo{year}{2010}).
	\newblock \bibinfo{title}{A new kernel-based approach for linear system
		identification}.
	\newblock {\it \bibinfo{journal}{Automatica}\/},  {\it \bibinfo{volume}{46}\/},
	\bibinfo{pages}{81--93}.
	\bibitem[{Pillonetto et~al.(2014)Pillonetto, Dinuzzo, Chen, De~Nicolao \&
		Ljung}]{Pillonetto2014}
	\bibinfo{author}{Pillonetto, G.}, \bibinfo{author}{Dinuzzo, F.},
	\bibinfo{author}{Chen, T.}, \bibinfo{author}{De~Nicolao, G.}, \&
	\bibinfo{author}{Ljung, L.} (\bibinfo{year}{2014}).
	\newblock \bibinfo{title}{Kernel methods in system identification, machine
		learning and function estimation: A survey}.
	\newblock {\it \bibinfo{journal}{Automatica}\/},  {\it \bibinfo{volume}{50}\/},
	\bibinfo{pages}{657--682}.
	\bibitem[{S\"oderstr\"om \& Stoica(1989)}]{Soderstrom1989}
	\bibinfo{author}{S\"oderstr\"om, T.}, \& \bibinfo{author}{Stoica, P.}
	(\bibinfo{year}{1989}).
	\newblock {\it \bibinfo{title}{System identification}\/}.
	\newblock \bibinfo{publisher}{Prentice Hall International}.
	\bibitem[{Tee(2007)}]{Tee2007}
	\bibinfo{author}{Tee, G.~J.} (\bibinfo{year}{2007}).
	\newblock \bibinfo{title}{Eigenvectors of block circulant and alternating
		circulant matrices}.
	\newblock {\it \bibinfo{journal}{New Zealand Journal of Mathematics}\/},  {\it
		\bibinfo{volume}{36}\/}, \bibinfo{pages}{195--211}.
	\bibitem[{Zarrop(1979)}]{Zarrop1977}
	\bibinfo{author}{Zarrop, M.~B.} (\bibinfo{year}{1979}).
	\newblock {\it \bibinfo{title}{Optimal experiment design for dynamic system
			identification}\/}.
	\newblock \bibinfo{address}{Berlin Heidelberg}:
	\bibinfo{publisher}{Springer-Verlag}.
	\bibitem[{Zorzi \& Chiuso(2017)}]{Zorzi2017}
	\bibinfo{author}{Zorzi, M.}, \& \bibinfo{author}{Chiuso, A.}
	(\bibinfo{year}{2017}).
	\newblock \bibinfo{title}{The harmonic analysis of kernel functions}.
	\newblock {\it \bibinfo{journal}{arXiv preprint arXiv:1703.05216}\/}, .
	
\end{thebibliography}

\end{document}